\documentclass[10pt]{article}
\usepackage{times}
\usepackage{amssymb}
\usepackage{epsfig}
\usepackage{xspace}
\usepackage{color}
\usepackage{epic}
\usepackage{eepic}

\usepackage{theorem}

\usepackage{epsf}
\usepackage{graphicx,pst-tree}
\usepackage{subfigure}
\usepackage{latexsym}
\newtheorem{theorem}{Theorem}[section]
{\theorembodyfont{\upshape}\newtheorem{example}[theorem]{Example}}

\newtheorem{definition}[theorem]{Definition}
\newtheorem{proposition}[theorem]{Proposition}

\newtheorem{corollary}[theorem]{Corollary}
\newtheorem{lemma}[theorem]{Lemma}

\newenvironment{proof}%
   {\begin{trivlist}\item[]\textbf{Proof.}}%
   {\end{trivlist}}

\newenvironment{proofsketch}%
  {\begin{trivlist}\item[]\textbf{Proof Sketch.}}%
  {\end{trivlist}}

\def\punto{\hspace*{\fill}\Box}

\newcommand{\nopAbduction}[1]{}
\newcommand{\nop}[1]{}

\def\Xplus{X^+}
\def\lhs{{\it lhs}}
\def\rhs{{\it rhs}}
\def\lh{{\it lh}}
\def\rh{{\it rh}}
\def\fd{{\it fd}}
\def\att{{\it att}}

\def\DC{\Delta C}
\def\FY{FY}
\def\FYone{FY_1}
\def\FYtwo{FY_2}
\def\FC{{\it FC}}
\def\C{C^o}
\def\At{{\it At}}
\def\Fd{{\it Fd}}

\def\consistent{{\it consistent}}
\def\outside{{\it outside}}
\def\unique{{\it unique}}

\def\FD{{\it FD}}
\def\Att{{\it Att}}

\def\YY{{\cal Y}}
\def\EY{\hat{Y}}
\def\EC{\hat{C}^o}
\def\EFC{\hat{FC}}

\def\ER{\hat{R}}
\def\EG{\hat{G}}
\def\EB{\hat{B}}

\def\primality{{\it prime}}

\def\RR{{\cal R}}
\def\GG{{\cal G}}

\def\AA{{\cal A}}
\def\BB{{\cal B}}
\def\EE{{\cal E}}
\def\RAi{R_i^\AA}
\def\RAone{R_1^\AA}
\def\RAk{R_K^\AA}

\def\tw{{\it tw}}

\def\dom{{\it dom}}

\def\root{{\it root}}
\def\leaf{{\it leaf}}
\def\child{{\it child}}
\def\bag{{\it bag}}

\def\tauTD{\tau_{td}}
\def\AATD{\AA_{td}}
\def\PP{{\cal P}}

\def\max{{\it max}}

\def\solve{{\it solve}}
\def\solveDx{{\it solve}$\downarrow$\xspace}
\def\solveD{{\it solve}$\downarrow$}

\def\succ{{\it success}}

\def\partition{{\it partition}}
\def\allowed{{\it allowed}}

\def\d-union{{\it disjoint\_union}}

\def\domA{{\it dom}({\cal A})}

\def\ThetaUp{\Theta^\uparrow}
\def\ThetaD{\Theta^\downarrow}
\def\bag{{\it bag}}
\def\child{{\it child}}

\def\mequiv{\equiv^{\it MSO}_k}
\def\A{{\cal A}}              % finite structure
\def\B{{\cal B}}              % finite structure
\def\S{{\cal S}}              % tree decomposition of $\A$
\def\T{{\cal T}}              % tree decomposition of $\B$
\def\I{{\cal I}}

\def\S{{\cal S}}
\def\R{{\cal R}}
\def\P{{\cal P}}              % monadic datalog program$
\def\ASs{{\langle {\cal A}, {\cal S}, s \rangle}}
\def\ASsp{{\langle {\cal A}, {\cal S}', s' \rangle}}
\def\ASspp{{\langle {\cal A}', {\cal S}', s' \rangle}}
\def\ATs{{\langle {\cal A}, {\cal T}, s \rangle}}
\def\BTt{{\langle {\cal B}, {\cal T}, t \rangle}}

  \newcommand{\ra}{\rightarrow}
  
  \newcommand{\la}{\leftarrow}
  
  \newcommand{\LR}{\Leftrightarrow}

  \renewcommand{\phi}{\varphi}
  \renewcommand{\theta}{\vartheta}

\def\TT{{\cal T}}
\def\Ts{{\cal T}_s}

\def\SS{{\cal S}}

\def\root{{\it root}}
\def\leaf{{\it leaf}}

\begin{document}

\title{Monadic Datalog over Finite Structures \\
 with Bounded Treewidth\thanks{This is an extended and enhanced version of results published in 
[19]. The work was partially supported by the Austrian Science Fund (FWF), project P20704-N18.}}

\author{
\parbox{6cm}{
\mbox{}
\hspace{-3cm}
\begin{tabular}[t]{c@{\extracolsep{1em}}c}
Georg Gottlob&
Reinhard Pichler\\
Computing Laboratory & 
Institut f\"{u}r Informationssysteme \\
Oxford University &
Technische Universit\"{a}t Wien\\
Oxford OX1 3QD,UK & 
A-1040 Vienna, Austria
\\
georg.gottlob@comlab.ox.ac.uk &
pichler@dbai.tuwien.ac.at 
\end{tabular} 
\\[1.1ex]
\begin{tabular}[t]{c}
Fang Wei\thanks{Work performed while the author was with Technische Universit\"{a}t Wien.} \\
Institut f\"{u}r Informatik \\
 Albert-Ludwigs-Universit\"{a}t Freiburg \\
D-79110 Freiburg i. Br., Germany \\
fwei@informatik.uni-freiburg.de 
\end{tabular}
}%parbox
}

\date{}

\maketitle

\begin{abstract}
Bounded treewidth and Monadic Second Order (MSO) logic 
have proved to be
key concepts in establishing fixed-para\-meter tractability
results. Indeed, by Courcelle's Theorem we know: 
Any property of finite structures, which is expressible by an MSO sentence, can be decided in linear time (data complexity) if the 
structures have bounded treewidth.
In principle, Courcelle's Theorem can be applied directly to 
construct concrete algorithms by transforming the MSO evaluation 
problem into a tree language recognition problem. The latter can then be 
solved via a finite tree automaton (FTA). However, this approach 
has turned out to be problematical, since even relatively simple MSO formulae may lead to a 
``state explosion'' of the FTA.

\hspace{2mm}
In this work we propose monadic datalog (i.e., datalog where all
intentional predicate symbols are unary) as an alternative method to
tackle this class of fixed-parameter tractable problems. 
We show that if some property of finite structures is expressible in MSO then
this property can also be expressed 
by means of a monadic datalog program over the structure plus the tree
decomposition. 
Moreover, we show that the resulting fragment of datalog
can be evaluated in linear time (both w.r.t.\  the program size and
w.r.t.\ the data size). This new approach is 
put to work by devising new algorithms for the 
3-Colorability problem of graphs and for the 
PRIMALITY problem of relational schemas 
(i.e., testing if
some attribute in a relational schema is
part of a key). We also report on experimental results 
with a prototype implementation.
\end{abstract}

\nop{****************
\vspace{1.5mm}
\noindent
{\bf Categories and Subject Descriptors:} 
%G.2.2 {[\bf Graph Theory]}: {Graph algorithms, Hypergraphs}
F.2.2 {\bf [Nonnumerical Algorithms and Problems]}: {
Complexity of proof procedures, Computations on discrete structures};
F.4.1 {\bf [Mathematical Logic]}: {
Computational Logic}

\vspace{1mm}
\noindent
{\bf General Terms:} Algorithms, Performance, Theory

\vspace{1mm}
}

%%%%%%%%%%%%%%%%%%%%%%%%%%%%%%%%%%%%%%%%%%%%%%%%%%%%%%%%%%%%%%
%%%%%%%%%%%%%% Introduction %%%%%%%%%%%%%%%%%%%%%%%%%%%%%%%%%%
\section{Introduction}
\label{sec:intro}

\nocite{GottlobPW07PODS}

Over the past decade, parameterized complexity has evol\-ved as 
an important subdiscipline in the field of computational complexity,
see \cite{DF99,FG06}.
In particular, it has been shown that many hard problems become tractable
if some problem parameter is fixed or bounded by a constant. In the
arena of graphs and, more generally, of finite structures, the treewidth is one such parameter which has served as the
key to many fixed-parameter tractability (FPT) results. 
The most prominent 
method for establishing the FPT in case of bounded treewidth is 
via Courcelle's Theorem, see \cite{Courcelle1990}: 
Any property of finite structures, which is expressible by a Monadic Second Order (MSO) sentence, can be decided in linear time (data complexity) if the treewidth of the structures is bounded by a fixed constant.

Recipes as to how one can devise concrete algorithms based on 
Courcelle's Theorem can be found in the literature, see \cite{ALS91,flum-frick-grohe-JACM02}. The idea is to 
first translate the MSO evaluation problem over finite structures
into an equivalent MSO evaluation problem over colored binary trees. 
This problem can then be solved via the correspondence between MSO
over trees and finite tree automata (FTA), see \cite{ThatcherW68,Doner70}.
In theory, this generic method of turning an MSO description into
a concrete algorithm looks very appealing. However, in practice, it
has turned out that even relatively simple MSO formulae may lead to a 
``state explosion'' of the FTA, see \cite{Frick02,Maryns06}. Consequently, it
was already stated in \cite{Grohe99} that the algorithms 
derived via Courcelle's Theorem are ``useless for practical applications''. The main benefit of Courcelle's Theorem is that it
provides ``a simple way to recognize a property as being 
linear time computable''. In other words, proving the 
FPT of some problem by showing that it is MSO expressible is 
the starting point (rather than the end point) of the search for an efficient algorithm.

In this work we propose monadic datalog (i.e., datalog where all
intensional predicate symbols are unary) as a practical tool
for devising efficient algorithms in situations where the FPT  
has been shown via Courcelle's Theorem. Above all, we prove
that if some property of finite structures is expressible in MSO then
this property can also be expressed 
by means of a monadic datalog program over the structure plus the tree
decomposition. 
Hence, in the first place, we prove 
an {\em expressivity result\/} rather than a mere complexity result. 
However, we also show that the resulting fragment of datalog
can be evaluated in linear time (both w.r.t.\ the program size and
w.r.t.\  the data size). We thus get the corresponding 
{\em complexity result\/} (i.e., 
Courcelle's Theorem) as a corollary of this MSO-to-datalog transformation.

Our MSO-to-datalog transformation for 
finite structures with 
bounded treewidth generalizes a result from 
\cite{GottlobK04} where it was shown that MSO on trees has the same
expressive power as monadic datalog on trees. 
\nop{Note that this 
generalization is by no means straightforward (just like Courcelle's Theorem is not a straightforward generalization of the results 
in \cite{ThatcherW68} and \cite{Doner70}).
}%nop
Several obstacles
had to be overcome to prove this generalization:
\begin{itemize}
\setlength{\itemsep}{-0.3mm}
\item First of all, we no longer have to deal with a single universe,
namely the universe of trees whose domain consists of the tree nodes.
Instead,
we now have to deal with -- and constantly switch between -- two universes, 
namely the relational structure (with its own signature and its own domain) on the one hand, and the tree decomposition (with appropriate predicates
expressing the tree structure and with the tree nodes as a separate domain) on the other hand.
\item Of course, not only the MSO-to-datalog transformation itself had to 
be lifted to the case of two universes. Also important prerequisites
of the results in \cite{GottlobK04} (notably several results on MSO-equivalences of tree structures shown in \cite{NevenS02})  had to be extended to this new situation.
\item Apart from switching between the two universes, it is
ultimately necessary to integrate both universes into the monadic
datalog program. For this purpose, both the signature and the domain
of the finite structure have to be appropriately extended.
\item It has turned out that previous notions of standard or  normal forms of tree decompositions (see \cite{DF99,flum-frick-grohe-JACM02}) are not suitable for our purposes. We therefore have to introduce a
modified version of ``normalized tree decompositions'', 
which is then further 
refined as we present new algorithms based on monadic datalog.
\end{itemize}

In the second part of this paper, we put monadic datalog to
work by presenting  
new algorithms for the
3-Colorability problem of graphs and for the 
PRIMALITY problem of relational schemas 
(i.e., testing if
some attribute in a relational schema is
part of a key). Both problems are well-known to be intractable 
(e.g., see \cite{mannila-DB-Design-Book} for PRIMALITY).
It is folklore that the 3-Colorability problem 
can be expressed by an MSO sentence.
In \cite{GottlobPW06PODS}, it was shown that PRIMALITY is
MSO expressible. 
Hence, in case of bounded treewidth, both problems  become tractable.
However, two attempts to tackle these problems via the standard MSO-to-FTA 
approach turned out to be very problematical:
We experimented with a prototype implementation using MONA (see \cite{mona}) for the MSO model checking, but 
we ended up with ``out-of-memory'' errors already for really small input data (see Section~\ref{sec:Results}).  
Alternatively, we made an attempt 
to directly implement the MSO-to-FTA mapping proposed in 
\cite{flum-frick-grohe-JACM02}. However, the
 ``state explosion'' of the resulting FTA -- which tends to occur already for comparatively simple formulae
(cf.\ \cite{Maryns06}) -- led to failure yet before we were able to feed any input data to the program.

In contrast, the experimental results with our 
new
datalog approach 
look
very promising, see Section~\ref{sec:Results}. 
By the experience gained with these experiments, 
the
following advantages of datalog compared with MSO became apparent:
\begin{itemize}
\setlength{\itemsep}{-0.3mm}
\item Level of declarativity. MSO as a logic has the 
highest level of declarativity which often allows one
very elegant and succinct problem specifications.
However, MSO does not have an
operational semantics. In order to turn an MSO specification
into an algorithm, the standard approach 
% in the literature 
is to transform the MSO evaluation problem into a tree language
recognition problem. 
But the FTA clearly has a much lower level of declarativity and 
the intuition of the original problem is usually lost
when an FTA is constructed.
In contrast, the datalog program with its declarative style often
reflects both the {\em intuition of the original problem and of 
the algorithmic solution}. 
This intuition can be exploited for defining heuristics
which lead to problem-speci\-fic
optimizations.

\item General optimizations. 
A lot of research 
has been devoted 
to generally applicable (i.e., not problem-speci\-fic) 
optimization techniques of datalog (see e.g.\ \cite{CGT90}). In our 
implementation (see Section~\ref{sec:Results}), 
we make heavy use of these optimization techniques, which are not
available in the MSO-to-FTA approach.
\item Flexibility. The generic transformation of MSO formulae to 
 monadic datalog programs (given in Section~\ref{sec:monDatalog})
inevitably leads to 
programs of 
exponential size w.r.t.\ the size of the MSO-formula and the treewidth.
However,
as our programs for  
3-Colorability and PRIMALITY 
demonstrate, ma\-ny relevant 
properties can be expressed by really short programs. Moreover, as we will see in Section~\ref{sec:Work},
also datalog provides us with a {\em certain level of succinctness}. In fact, we
will be able to express a big monadic datalog program 
by a small non-monadic program.

\item Required transformations. The problem of a
``state explosion'' reported in \cite{Maryns06}
already refers to the transformation of (relatively simple) 
MSO formulae {\em on trees\/} to an FTA. 
If we consider MSO {\em on structures with bounded
treewidth\/} the situation gets even worse, since the original
(possibly simple) MSO formula over a finite structure first has to 
be transformed into an equivalent MSO formula over trees. This transformation
(e.g., by the algorithm in \cite{flum-frick-grohe-JACM02}) leads to a much 
more com\-plex for\-mula (in ge\-ne\-ral, even with additio\-nal quantifier alternations)
than the original formula.
In contrast, our approach works with 
monadic datalog programs on finite structures which need no further transformation. Each program can be executed as it is.

\item Extending the programming language. 
One more aspect of the flexibility
of datalog is the possibility to define new
built-in predicates if they admit an 
efficient implementation by the 
interpreter. Another example of a useful language extension is the
introduction of generalized quantifiers. For the theoretical background
of this concept, see \cite{EiterGV97b,EiterGV97a}.
\end{itemize}
Some applications require a fast execution which cannot always be
guaranteed by an interpreter. Hence, while
we propose a logic programming approach, one can of course go 
one step further and implement our algorithms directly in Java, C++, etc. 
following the same paradigm.

The paper is organized as follows. After recalling some basic notions and
results in Section~\ref{sec:preliminaries}, we prove several
results on the MSO-equivalence of substructures induced by subtrees of a
tree decomposition in Section~\ref{sec:generalization-NS}. In 
Section~\ref{sec:monDatalog}, it is shown that any MSO formula with 
one free individual variable over structures with bounded treewidth can be transformed into an equivalent monadic datalog program. 
% Moreover, we show that the resulting fragment of datalog can be evaluated in
% linear time both w.r.t.\ the size of the program and 
% w.r.t.\ the size of the data. 
In Section~\ref{sec:Work}, we put monadic datalog to work by 
presenting new FPT algorithms for the 
3-Colorability problem and for the 
PRIMALITY problem
in case of bounded treewidth. 
In Section~\ref{sec:Results}, we report on experimental results with
a prototype implementation. 
A conclusion is given in
Section~\ref{sec:Conclusion}.

%%%%%%%%%%%%%%%%%%%%%%%%%%%%%%%%%%%%%%%%%%%%%%%%%%%%%%%%%%%%%%
%%%%%%%%%%%%%% Preliminaries %%%%%%%%%%%%%%%%%%%%%%%%%%%%%%%%%%
\section{Preliminaries}
\label{sec:preliminaries}

\subsection{Relational Schemas and Primality}
\label{sec:relational}

We briefly recall some basic notions and results from 
database design theory (for details, see 
\cite{mannila-DB-Design-Book}).
In particular, we shall define the PRIMALITY problem, which will serve as a running example throughout this paper.

A relational schema is denoted as $(R,F)$ where 
$R$ is the set of attributes, and $F$ the set of functional dependencies (FDs, for short) over $R$. W.l.o.g., we only consider FDs whose
right-hand side consists of 
a single attribute. Let $f\in F$ with $f\colon Y \ra A$. We refer
to $Y \subseteq R$ and $A \in R$ as $\lhs(f)$ and $\rhs(f)$, 
respectively. The intended meaning of an FD $f\colon Y \ra A$ is that, 
in any valid database instance of $(R,F)$, the value of the attribute $A$ is
uniquely determined by the value of the attributes in $Y$. It is convenient to
denote a set $\{A_1, A_2, \dots, A_n\}$ of attributes as a string
$A_1 A_2 \dots A_n$. For instance, we write $f\colon ab \ra c$ rather than 
$f\colon \{a, b\} \ra c$.

For any $X \subseteq R$, we 
write $\Xplus$ to denote the closure of $X$, i.e., the set
of all attributes determined by $X$. An attribute
$A$ is contained in $\Xplus$ iff either $A \in X$ or there 
exists a ``derivation sequence''  of $A$ from $X$ in $F$
of the form
$X \ra  X \cup \{A_{1} \} \ra X \cup \{A_{1}, A_{2} \} \ra
\dots \ra X \cup \{A_{1}, \dots, A_{n}  \}$, s.t.\
$A_n = A$ and for every $i \in \{1, \dots, n\}$, there
exists an FD $f_i \in F$ with 
$\lhs(f) \subseteq X \cup \{A_{1}, \dots, A_{i-1}  \}$ and
$\rhs(f) = A_i$. 

If $\Xplus = R$ then $X$ is called a \emph{superkey}.
If $X$ is minimal with this property, then $X$ is a \emph{key}. 
An attribute $A$ is called {\em prime}
%  in $(R,F)$, 
if it is contained in at least one key in $(R,F)$. 
An efficient algorithm for testing the primality 
of an attribute is crucial in 
database design since it is an indispensable prerequisite
for testing if a schema is in third normal form.
However, given a relational schema $(R,F)$ and an attribute
$A \in R$, it is NP-complete to test if $A$ is prime 
(cf.\ \cite{mannila-DB-Design-Book}).

We shall consider two variants of the PRIMALITY problem
in this paper (see Section \ref{sec:primality} and
\ref{sec:monadic-prime}, resp.): the decision problem (i.e,
given a relational schema $(R,F)$ and an attribute $A \in R$, is $A$ prime in $(R,F)$?) and the enumeration problem
(i.e,
given a relational schema $(R,F)$, compute all 
prime attributes in $(R,F)$).

\begin{example}
\label{bsp:rel-schema}
Consider the relational schema $(R,F)$ with $R = abcdeg$
and $F = \{ f_1\colon  ab \ra c, f_2\colon c \ra b$, $f_3\colon cd \ra e, 
f_4\colon de \ra g, f_5\colon g \ra e\}$.
It can be easily checked that there are two keys for the schema: $abd$ and $acd$.
Thus the attributes $a,b,c$ and $d$ are prime, while $e$ and $g$ are not prime.

\nop{Ein Beispiel fuer ein relationales Schema und alle prime attributes in diesem Schema angeben. Dabei die Attribute und FDs mit Kleinbuchstaben benennen, damit die Notation mit mit dem 
Programm konsistent ist
}
\end{example}

\subsection{Finite Structures and Treewidth}
\label{sec:treewidth}

Let $\tau = \{R_1, \dots, R_K\}$ be a set of predicate symbols.
A {\em finite structure\/}
 $\AA$ over $\tau$ (a {\em $\tau$-structure\/}, for short) 
is given by a finite domain $A = \domA$ and relations 
$\RAi \subseteq A^\alpha$, where $\alpha$ denotes the arity
of $R_i \in \tau$. A finite structure may also
be given in the form $(\AA, \bar{a})$ where, in addition to  $\AA$, 
we have distinguished elements $\bar{a} = (a_0, \dots, a_w)$ from 
$\domA$.
% the  domain $A$. 
Such distinguished elements are required for interpreting
 formulae with free variables.

A {\em tree decomposition\/} $\TT$ of a $\tau$-structure $\AA$
is defined as a pair $\langle T, (A_t)_{t\in T} \rangle$ where
$T$ is a tree and each $A_t$ is a subset of $A$
with the following properties: (1) Every $a \in A$ is contained
in some $A_t$. (2) For every $R_i \in \tau$ and every tuple $(a_1, \dots, a_\alpha) \in \RAi$, there exists some node $t \in T$
with $\{a_1, \dots, a_\alpha \} \subseteq A_t$. (3) For every 
$a \in A$, the set $\{t \mid a \in A_t\}$ induces a subtree of $T$.

The third condition is usually referred to as the {\em connectedness
condition}. The sets $A_t$ are called the {\em bags\/} 
(or {\em blocks\/}) of $\TT$.
The {\em width\/} of a tree decomposition
$\langle T, (A_t)_{t\in T} \rangle$ is defined as
$\max \{ |A_t| \mid t \in T \} - 1$. The {\em treewidth\/} of
$\AA$ is the minimal width of all tree decompositions of 
$\AA$. It is denoted as $\tw(\AA)$. Note that trees and forests
are precisely the structures with treewidth 1.

For given $w \geq 1$, it can be decided in linear time 
if some structure has treewidth $\leq w$. Moreover, in case of a positive 
answer, a tree decomposition of width $w$ can be computed in 
linear time, see  \cite{Bod96}.

In this paper, we assume that a relational schema $(R,F)$
is given as a $\tau$-structure with 
$\tau = \{\fd, \att, \lh$, $\rh
 \}$.
The intended meaning of these predicates is as follows: 
$\fd(f)$ means that $f$ is an FD and 
$\att(b)$ means that $b$ is an attribute. 
$\lh(b,f)$ (resp.\ $\rh(b,f)$) 
means that $b$ occurs in $\lhs(f)$ (resp.\ in $\rhs(f)$).
The treewidth of $(R,F)$ is then defined as the treewidth of 
this $\tau$-structure.

\begin{example}
\label{bsp:tree-decomposition}
Recall the 
relational schema $(R,F)$ with $R = abcdeg$
and $F = \{ f_1\colon  ab \ra c, f_2\colon c \ra b$, $f_3\colon cd \ra e, 
f_4\colon de \ra g, f_5\colon g \ra e\}$ from 
Example \ref{bsp:rel-schema}. This schema
is represented as the following $\tau$-structure with 
$\tau = \{\fd, \att, \lh, \rh
 \}$: ${\cal A} = (A, \fd^{\cal A}, 
\att^{\cal A}, \lh^{\cal A}, \rh^{\cal A})$
with $A = R$, 
$\fd^{\cal A} = \{f_1, f_2, f_3, f_4, f_5\}$, 
$\att^{\cal A}= \{a,b,c,d,e,g\}$,
$ \lh^{\cal A}= \{(a,f_1), (b,f_1), (c,f_2), (c,f_3), (d,f_3), (d,f_4), (e,f_4), (g,f_5)\}$, 
$\rh^{\cal A} = \{(c,f_1), (b,f_2)$, $(e,f_3), (g,f_4), (e,f_5)$.

A tree decomposition 
${\cal T}$
of this structure is given in Figure~\ref{fig:tree-decomp}.
Note that the maximal size of the bags in $\TT$ is $3$. Hence, the tree-width is
$\leq 2$.
On the other hand, 
it is easy to check that the tree-width of $\TT$ cannot be smaller than $2$: 
In order to see this, we consider the tuples in 
$ \lh^{\cal A}$ and $\rh^{\cal A}$ as edges of an undirected graph. Then the 
edges corresponding to 
$(b,f_1), (c,f_2) \in  \lh^{\cal A}$ and 
$(b,f_2), (c,f_1) \in  \rh^{\cal A}$ form a cycle in this graph. However, 
as we have recalled
above, only trees and forests have treewidth 1. 
The tree decomposition in Figure~\ref{fig:tree-decomp} is, therefore, optimal and we have
$\tw(F) = \tw(\AA) = 2$.

\nop{Relational schema als $\tau$-structure darstellen und Grafik mit einer (nicht normalisierten) Tree decomposition ${\cal T}$ angeben. Idealerweise im Beispiel dann erlaeutern, wieso die width ${\cal T}$ optimal ist und wir daher die treewidth von $(R,F)$ bekommen.} 
\end{example}

\begin{figure}[h]
\begin{center}
\includegraphics[scale=0.4]{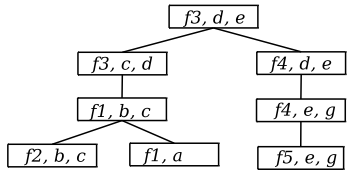}
\end{center}
\vspace{-0.5cm}
\caption{Tree decomposition ${\cal T}$ of 
schema $(R,F)$ in Example~\ref{bsp:rel-schema}}
\label{fig:tree-decomp}
\end{figure}

{\em Remark.}
A relational schema $(R,F)$ defines a hypergraph $H(R,F)$ 
whose vertices are the attributes in  $R$ and whose hyperedges 
are the sets of attributes jointly occurring in at least one
FD in $F$. Recall that the incidence graph of a hypergraph $H$
contains as nodes the vertices and hyperedges of $H$.
Moreover,
two nodes $v$ and $h$ (corresponding to a vertex $v$ and a 
hyperedge $h$ in $H$) 
are connected in this graph iff (in the hypergraph $H$) 
$v$ occurs in $h$.
It can be easily verified that the treewidth of the above described 
$\tau$-structure and 
of the incidence graph of the hypergraph $H(R,F)$ coincide.

In this paper, we consider the following form of 
{\em normalized tree decompositions\/}, which is
similar to the normal form introduced in Theorem 6.72 of 
\cite{DF99}:

\begin{definition}
\label{def:normalized}
Let  $\AA$ be an arbitry structure with 
tree decomposition $\TT$ of width $w$. 
We call $\TT$ {\em normalized\/} if the conditions 1 -- 4 are fulfilled:
(1) The bags are considered
as tuples of $w+1$ pairwise distinct elements $(a_0, \dots, a_w)$ rather than sets. 
(2) Every internal node $t \in T$ has 
either 1 or 2 child nodes. (3) If a node $t$ with bag
$(a_0, \dots, a_w)$ 
has one child node, then the bag of the child is either obtained
via a permutation of $(a_0, \dots, a_w)$  or by replacing
$a_0$ with another element $a_0'$. We call such a node $t$ 
a {\em permutation node\/} or an {\em element replacement node\/},
respectively.
(4) If a node $t$ has two child nodes then these child nodes have 
identical bags as $t$.  In this case, we call  $t$
a {\em branch node\/}. 
\end{definition}

\begin{proposition}
Let  $\AA$ be an arbitry structure with 
tree decomposition $\TT$ of width $w$. 
W.l.o.g., we may assume that the domain $\dom(\AA)$ has at least $w+1$ elements.
Then $\TT$ can be transformed in linear time into a normalized tree decomposition $\TT'$, s.t.\ $\TT$ and $\TT'$ have identical width.
\end{proposition}

\begin{proof}
We can 
transform an arbitrary tree decomposition $\TT$ into
a normalized tree decomposition $\TT'$
by the following steps (1) - (5). Clearly 
this transformation works in in linear time
and preserves the width.

(1) All bags can be padded to the ``full'' size of $w+1$ elements by adding 
elements from a neighboring bag, e.g.: Let $s$ and $s'$
be adjacent nodes and let $A_s$ have $w+1$ elements (in a tree decomposition of 
width $w$, at least one such node exists) and let $|A_{s'}| = w' + 1$ with
$w' < w$. Then $|A_s \setminus A_{s'}| \geq (w-w')$ and we may
simply add $(w-w')$ elements from $A_s \setminus A_{s'}$ to $A_{s'}$
without violating the 
connectedness condition. 
% of tree decompositions 
% is not violated by this transformation.

(2) Suppose that some internal node $s$ has $k+2$ child nodes $t_1, \dots, t_{k+2}$
with $k > 0$. It is a standard technique to turn this part of the tree into 
a binary tree by inserting copies of $s$ into the tree, i.e., we introduce $k$ 
nodes $s_1, \dots, s_{k}$ with $A_{s_i} = A_s$, s.t.\
the second child of $s$ is $s_1$, the second child of
$s_1$ is $s_2$, the second child of $s_2$ is $s_3$, etc. 
Moreover, $t_1$ remains the first child of $s$, while $t_2$ becomes the 
first child of $s_1$, $t_3$ becomes the 
first child of $s_2$, \dots, 
$t_{k+1}$ becomes the 
first child of $s_k$. Finally, $t_{k+2}$ becomes the second child of 
$s_{k}$. Clearly, the connectedness condition is preserved by this construction.

(3) If an internal node $s$ has two children $t_1$ and $t_2$, s.t.\  the
bags of $s$, $t_1$, and $t_2$ are not identical, then 
we simply insert a copy $s_1$ of $s$ between $s$ and $t_1$ and another
copy $s_2$ of $s$ between $s$ and $t_2$.

(4) Let $s$ be the parent of  $s'$ and let
 $|A_s \setminus A_{s'}| = k$
with $k > 1$. Then we can obviously ``interpolate'' $s$ and $s'$ by
new nodes $s_1, \dots, s_{k-1}$, s.t.\ $s_{k-1}$ is the new parent of $s'$, 
$s_{k-2}$ is the parent of $s_{k-1}$, \dots, $s$ is the parent of $s_1$.
Moreover, the bags $A_{s_i}$ can be defined in such a way that the bags of 
any two neighboring nodes differ in exactly one element, e.g. 
$|A_s \setminus A_{s_1}| =  |A_{s_1} \setminus A_{s}| = 1$. 

(5) Let the bags of any two neighboring nodes $s$ and $s'$ differ by  one element, i.e., $\exists a \in A_{s}$ with $a \not \in A_{s'}$ and
$\exists a' \in A_{s'}$ with $a' \not \in A_{s}$. Then we can insert two 
``interpolation nodes'' $t$ and $t'$, s.t.\  $A_t$ has the same elements as
$A_s$ but with $a$  at position $0$. Likewise, $A_{t'}$ has the same elements as
$A_{s'}$ but with  $a'$  at position $0$.
$\punto$
\end{proof}

\begin{example}
\label{bsp:normalized-tree-decomposition}
The tree decomposition ${\cal T}$ in Figure~\ref{fig:tree-decomp}
is clearly not normalized. In contrast, 
tree decomposition ${\cal T}'$ in Figure~\ref{fig:normalized-tree-decomp}
is normalized in the above sense. Let us ignore the node identifiers $s_1,\ldots,s_{22}$ for the moment. Note the 
${\cal T}$ and ${\cal T}'$ have identical
width.
\nop{Grafik mit einer normalisierten Tree decomposition ${\cal T}'$ angeben. } 
\end{example}

\begin{figure}[h]
\begin{center}
\includegraphics[scale=0.4]{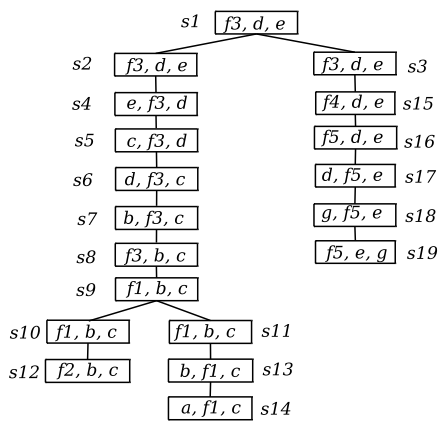}
\end{center}
\vspace{-0.5cm}
\caption{Normalized tree decomposition ${\cal T}'$ of 
schema $(R,F)$ in Example~\ref{bsp:rel-schema}}
\label{fig:normalized-tree-decomp}
\end{figure}

\subsection{Monadic Second Order Logic}
\label{sec:MSO}

We assume some familiarity with Monadic
Second Order logic (MSO), see e.g.\ 
\cite{EF99,Libkin04}.
MSO extends First Order logic (FO) by the use
of {\em set variables\/} (usually denoted by upper case letters), which range over sets of domain elements.
In contrast, the {\em individual variables\/} (which are usually denoted by
lower case letters)
range over single domain elements.
An FO-formula $\phi$ over a $\tau$-structure has as atomic formulae
either atoms with some predicate symbol from $\tau$ or equality atoms. 
An MSO-formula $\phi$ over a $\tau$-structure may additionally have
atoms whose predicate symbol is a monadic predicate variable. 
For the sake of readability, we denote such an atom usually as $a \in X$ rather than 
$X(a)$. Likewise, we use set operators $\subseteq$ and  $\subset$ with the obvious 
meaning. 

The {\em quantifier depth\/} of an MSO-formula $\phi$ is 
defined as the maximum degree of 
nesting of quantifiers (both for individual variables and set variables) 
in $\phi$.
In this work, we will mainly encounter MSO formulae with free individual
variables. A formula $\phi(x)$ with exactly one free individual variable
is called a {\em unary query\/}. 
More generally, let $\phi(\bar{x})$ with $\bar{x} = (x_0, \dots, x_w)$ for some 
$w \geq 0$ be an MSO formula with free variables $\bar{x}$.
Furthermore, let $\AA$ be a $\tau$-structure and 
$\bar{a} = (a_0, \dots, a_w)$ be distinguished domain elements. 
We write $(\AA,\bar{a}) \models \phi(\bar{x})$ to denote that
$\phi(\bar{a})$ evaluates to true in $\AA$. 
Usually, we
refer to $(\AA,\bar{a})$ simply as a ``structure'' rather than 
a ``structure with distinguished domain elements''.

\begin{example}
\label{bsp:primality-MSO}
It was shown in \cite{GottlobPW06PODS} that primality can be expressed in MSO. We give a slightly different MSO-formula $\phi(x)$  here, which is better suited for our purposes in Section~\ref{sec:Work}, namely

\smallskip

$\phi(x) = (\exists Y) [Y \subseteq R \wedge Closed (Y) \wedge x \not\in Y \wedge
Closure (Y \cup \{x\},R)]$ with 

$Closed(Y) \equiv 
(\forall f) [fd(f) \ra (\exists b)[
(\rh(b,f) \wedge b \in Y) \vee 
(\lh(b,f) \wedge b \not\in Y)]]$
and

$Closure (Y,Z) \equiv 
Y \subseteq Z \wedge Closed(Z) \wedge \neg (\exists Z') 
[Y \subseteq Z' \wedge Z' \subset Z \wedge
Closed(Z')]$.

\smallskip

\noindent
This formula expresses the following characterization of 
primality: 
An attribute $a$ is prime, iff there exists 
an attribute set $\YY \subseteq R$, s.t.\ $\YY$ is closed
w.r.t.\ $F$ (i.e., $\YY^+ = \YY$), $a \not\in \YY$ and
$(\YY \cup \{a\})^+ = R$. In other words, 
$\YY \cup \{a\}$ is a superkey but $\YY$ is not. 

\nop{Ergaenze Beispiel aus Example \ref{bsp:tree-decomposition}, z.B.:}

Recall the $\tau$-structure $\AA$ from Example
\ref{bsp:tree-decomposition} representing a 
relational schema. It can be 
easily verified that 
$(\AA, a) \models \phi(x)$ and
$(\AA, e) \not\models \phi(x)$ hold.
\end{example}

We call two structures $(\AA, \bar{a})$ and $(\BB, \bar{b})$
{\em $k$-equivalent\/} and write
$(\AA, \bar{a}) \mequiv (\BB, \bar{b})$,
iff for every MSO-formula $\phi$ of 
quantifier depth $\leq k$, the equivalence
$(\AA,\bar{a}) \models \phi \LR (\BB,\bar{b}) \models \phi$ holds.
By definition, $\mequiv$ is an equivalence relation. For any $k$, 
the relation $\mequiv$ has only finitely many equivalence classes. 
These equivalence classes are referred to as
{\em $k$-types} or simply as {\em types}.
The $\mequiv$-equivalence between two structures can be
effectively decided. 
There is a nice characterization 
of $\mequiv$-equivalence 
by Ehrenfeucht-Fra\"{\i}ss\'{e} games:
The 
{\em $k$-round MSO-game} 
on two structures $(\AA, \bar{a})$ and $(\BB, \bar{b})$ is
played between two players -- the spoiler and the duplicator. 
In each of the $k$ rounds, the spoiler can choose between 
a point move and a set move. If, in the $i$-th round, he makes a 
{\em point move}, then he selects some element $c_i \in \dom(\AA)$ or
some element $d_i \in \dom(\BB)$. The duplicator answers by choosing
an element in the opposite structure. If, in the $i$-th round, 
the spoiler makes a {\em set move}, 
then he selects a set $P_i \subseteq \dom(\AA)$ or
a set $Q_i \subseteq \dom(\BB)$.
The duplicator answers by choosing a set of domain 
elements in the opposite structure. 
Suppose that, in $k$ rounds, the domain elements
$c_1, \dots, c_m$ and $d_1, \dots, d_m$ from $\dom(\AA)$ 
and $\dom(\BB)$, respectively, were chosen in the point moves. Likewise, 
suppose that the subsets
$P_1, \dots, P_n$ and $Q_1, \dots, Q_m$ of $\dom(\AA)$ 
and $\dom(\BB)$, respectively, were chosen in the set moves. 
The {\em duplicator wins} this game, if the mapping which maps each $c_i$ to
$d_i$ is a partial isomorphism from 
$(\AA, \bar{a}, P_1, \dots, P_n)$
to $(\BB, \bar{b}, Q_1, \dots, Q_n)$.
We say that the duplicator has a {\em winning strategy}
in the $k$-round MSO-game on 
$(\AA, \bar{a})$ and $(\BB, \bar{b})$
if he can win the game for any possible moves
of the spoiler.

The following relationship between $\mequiv$-equivalence and
$k$-round MSO-games holds: {\em Two structures $(\AA, \bar{a})$ and $(\BB, \bar{b})$ are $k$-equivalent iff
the duplicator has a winning strategy in the $k$-round MSO-game on 
$(\AA, \bar{a})$ and $(\BB, \bar{b})$\/}, see \cite{EF99,Libkin04}.

\subsection{Datalog}
\label{sec:Datalog}

We assume some familiarity with datalog, see e.g.\ 
\cite{AHV95,CGT90, Ull89}.
Syntactically, a datalog program $\P$ is a set
of function-free Horn clauses. The (minimal-model) semantics  
can be defined as the least fixpoint of applying the immediate
consequence operator. Predicates occurring only in the body of 
rules in $\P$ are called extensional, while predicates occurring
also in the head of some rule are called intensional.

Let $\AA$ be a $\tau$-structure with domain $A$ and
relations $\RAone$, $\dots, \RAk$ with 
$\RAi \subseteq A^\alpha$, where $\alpha$ denotes the arity
of $R_i \in \tau$. In the context of datalog, it is convenient to 
think of the relations $\RAi$ as sets of ground atoms. The 
set of all such ground atoms of a structure $\AA$ is referred to as
the extensional database (EDB) of $\AA$, 
which we shall denote as $\EE(\AA)$ (or simply as $\AA$, if no confusion
is possible). We have
$R_i(\bar{a}) \in \EE(\AA)$ iff $\bar{a} \in \RAi$.

Evaluating a datalog program $\P$ over a structure $\AA$ comes down
to computing the least fixpoint of $\P \cup \AA$. Concerning the 
complexity of datalog, we are mainly interested in the combined
complexity (i.e., the complexity w.r.t.\ the size of the program $\P$
plus the size of the data $\AA$). 
In general, the combined complexity of datalog is EXPTIME-complete (implicit
in \cite{Vardi82}). However, 
there are some fragments which can be evaluated much more efficiently. 
(1) {\em Propositional datalog\/} (i.e., all rules are ground) can be evaluated
in linear time (combined complexity), see \cite{DowlingG84,Minoux88}.
(2) The {\em guarded fragment\/} of datalog
(i.e., every rule $r$ contains an extensional atom $B$ in the body, 
s.t.\ all variables occurring 
in $r$ also occur in $B$) can be evaluated in time ${\cal O}(|\P| * |\AA|)$.
(3) {\em Monadic datalog\/} (i.e., all intensional predicates are unary) is NP-complete (combined complexity), see \cite{GottlobK04}.

\section{Induced substructures}
\label{sec:generalization-NS}

In this section, we study the $k$-types of 
substructures induced by certain subtrees of a tree decomposition 
(see Definitions~\ref{def:subtree-envelope} and \ref{def:induced}).
Moreover, 
it is convenient to introduce 
some additional notation in 
Definition~\ref{def:equivalent-tuples}
below.
%& The subtrees and substructures 
% we are interested in are of the following form:

% thus generalizing results from \cite{NevenS02}.
%

\begin{definition}
\label{def:subtree-envelope}
Let $T$ be a tree and $t$ a node in $T$. Then we denote the {\em subtree rooted at $t$} as $T_t$. Moreover, analogously to \cite{NevenS02}, we write $\bar{T}_t$ to denote the 
{\em envelope of $T_t$}. This envelope is obtained by removing all of $T_t$ \,from $T$ except for the node $t$. 

Likewise, let $\T = \langle T, (A_s)_{s\in T} \rangle$ be a tree decomposition
of a finite structure. Then we define 
$\T_t = \langle T_t, (A_s)_{s\in T_t} \rangle$ and
$\bar{\T}_t = \langle \bar{T}_t, (A_s)_{s\in \bar{T}_t} \rangle$.
\end{definition}
In other words, $t$ is the root
node in $T_t$ while, in $\bar{T}_t$, it is a leaf node. Clearly, 
the only node occurring in both $T_t$ and $\bar{T}_t$ is $t$.

\begin{definition}
\label{def:induced}
Let $\A$ be a finite structure and let
$\T = \langle T, (A_t)_{t\in T} \rangle$
be a tree decomposition of $\A$.
Moreover, let 
$s$ be a node in $\T$ with bag 
$A_s = \bar{a} = (a_0, \dots, a_w)$
and let $\S$ be one of the subtrees 
$\T_s$ or $\bar{\T}_s$ of $\T$.

Then we write $\I(\AA, \S, s)$ to denote the 
structure $(\AA', \bar{a})$, 
where $\AA'$ is the substructure of $\A$ induced by the elements occurring in the bags of $\S$ .
\end{definition}

\begin{example}
\label{bsp:induced-substructure}

Recall the 
relational schema $(R,F)$ represented by the structure $\AA$ from Example \ref{bsp:tree-decomposition}
with (normalized) tree decomposition $\TT'$ in
Figure~\ref{fig:normalized-tree-decomp}. Consider,
for instance, the node $s$ in $\TT'$, as depicted in Figure \ref{fig:normalized-tree-decomp-induced},
with bag $A_s = (b,c)$.
Then the induced substructure 
$\I(\AA, \T'_s, s)$ is the substructure of $\AA$ which is induced by the elements occurring in the bags of $\T'_s$,
whereas $\I(\AA, \bar{\T}'_s, s)$ the substructure of $\AA$ which is induced by the elements occurring in the bags of $ \bar{\T}'_s$.
\nop{Irgendeinen Knoten aus der Tree decomposition $\TT'$ aussuchen und die beiden induced substructures
$\I(\AA, \T'_s, s)$ und
$\I(\AA, \bar{\T}'_s, s)$ angeben.}
\end{example}

\begin{figure}[h]
\begin{center}
\includegraphics[scale=0.4]{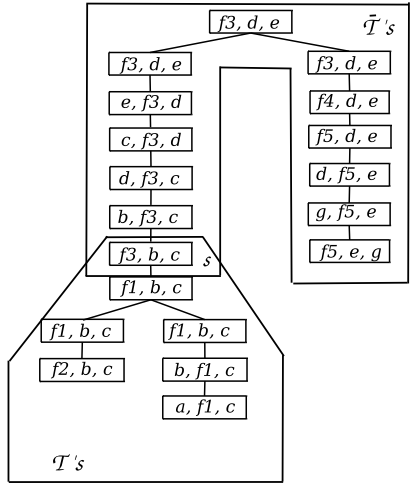}
\end{center}
\vspace{-0.5cm}
\caption{Induced substructures $\T'_s$ and $\bar{\T}'_s$ of the tree decomposition ${\cal T}$ w.r.t. the node $s$.}
\label{fig:normalized-tree-decomp-induced}
\end{figure}

\begin{definition}
\label{def:equivalent-tuples}
Let $w \geq 1$ be a natural number and let
$\A$ and $\B$ be finite structures over some signature $\tau$. 
Moreover, let $(a_0, \dots, a_w)$ (resp.\ $(b_0, \dots, b_w)$) be 
a tuple of pairwise distinct elements in $A$ (resp.\ $B$). 

We call
$(a_0, \dots, a_w)$ and $(b_0, \dots, b_w)$ 
{\em equivalent\/} and write
$(a_0, \dots, a_w) \equiv (b_0, \dots, b_w)$, iff for any predicate
symbol $R \in \tau$ with arity $\alpha$ and for
all tuples $(i_1, \dots, i_\alpha) \in \{0, \dots, w\}^\alpha$, the 
equivalence $R^\A(a_{i_1}, \dots, a_{i_\alpha})$ 
$\LR$ $R^\B(b_{i_1}, \dots, b_{i_\alpha})$ holds.
\end{definition}

We are now ready to generalize results
%Proposition 2.9 
from \cite{NevenS02} (dealing with trees plus a distinguished node) 
to the case of finite structures of bounded treewidth over an arbitrary signature $\tau$. In the three lemmas below, let $k \geq 0$ and $w \geq 1$ be arbitrary natural numbers and let $\tau$ be an arbitrary signature.

\begin{lemma}
% [bottom-up traversal of the tree decomposition]
\label{lem:folklore-bottomUp}
Let $\A$ and $\B$ be $\tau$-structures, let
$\S$ (resp.\ $\T$) be a normalized tree decomposition of $\A$ (resp.\ of $\B$)
of width $w$, and let $s$ (resp.\ $t$) be an internal node in $\S$ (resp.\ in $\T$).

\smallskip

\noindent
(1) {\em permutation nodes.} Let $s'$ (resp.\ $t'$) be the only child of $s$ in $\S$ (resp.\ of $t$ in $\T$). Moreover, let 
% $(a_0, \dots, a_w)$, $(a'_0, \dots, a'_w)$, 
% $(b_0, \dots, b_w)$, and $(b'_0, \dots, b'_w)$
$\bar{a}$, $\bar{a}'$, $\bar{b}$, and $\bar{b}'$ 
denote the bags at the nodes $s$, $s'$, $t$, and $t'$, respectively.

\smallskip

\noindent
If $\I(\AA,\S_{s'},s') \mequiv \I(\BB, \T_{t'},t')$ and there exists a permutation
$\pi$, s.t.\ 
$\bar{a} = \pi(\bar{a}')$ and
$\bar{b} = \pi(\bar{b}')$ 

then
$\I(\AA, \S_{s},s) \mequiv \I(\BB, \T_{t},t)$.

\medskip

\noindent
(2) {\em element replacement nodes.} Let $s'$ (resp.\ $t'$) be the 
only child  of $s$ in $\S$ (resp.\ of $t$ in $\T$). Moreover, let 
$\bar{a} = (a_0, a_1$, $\dots, a_w)$, 
$\bar{a}' = (a'_0, a_1, \dots, a_w)$, 
$\bar{b} = (b_0, b_1,$ $\dots, b_w)$, 
and $\bar{b}' = (b'_0, b_1, \dots, b_w)$
denote the bags at the nodes $s$, $s'$, $t$, and $t'$, respectively.

\smallskip

\noindent
If $\I(\AA,\S_{s'},s') \mequiv 
\I(\BB,\T_{t'},t') $ and 
$\bar{a}  \equiv \bar{b}$ then
$\I(\AA,\S_{s},s)$ $\mequiv 
\I(\BB,\T_{t},t) $.

\medskip

\noindent
(3) {\em branch nodes.} Let $s_1$ and $s_2$  (resp.\ $t_1$ and $t_2$) be the children of $s$ in $\S$ (resp.\ of $t$ in $\T$). 

\smallskip

\noindent
If 
$\I(\AA,\S_{s_1},s_1) \mequiv \I(\BB,\T_{t_1},t_1)$ 
and 
$\I(\AA,\S_{s_2},s_2)$ $\mequiv \I(\BB,\T_{t_2},t_2)$ 

then 
$\I(\AA,\S_{s},s) \mequiv \I(\BB,\T_{t},t)$.
\end{lemma}

\begin{proof}
\mbox{}
%\smallskip

\noindent
(1) Let $\I(\AA, \S_{s'}, s') \mequiv \I(\BB, \T_{t'}, t')$.
Hence, there exists 
a winning strategy of the duplicator on these structures. Moreover, 
$(a_0, \dots, a_w)$ and $(b_0, \dots, b_w)$ are obtained from 
$(a'_0, \dots, a'_w)$ resp.\ $(b'_0, \dots, b'_w)$ by identical permutations.
Thus the duplicator's winning strategy on the structures 
$\I(\AA, \S_{s'}, s')$ and $\I(\BB, \T_{t'}$, $t')$
is also a winning strategy on 
$\I(\AA, \S_{s}, s)$ and $\I(\BB, \T_{t}, t)$.

\medskip

\noindent
(2) Let 
$\I(\AA, \S_{s'}, s') \mequiv \I(\BB, \T_{t'}, t')$.
Hence, there exists 
a winning strategy of the duplicator on these structures. The duplicator 
extends this strategy to the structures 
$\I(\AA, \S_{s}, s)$ and $\I(\BB, \T_{t}, t)$
in the following way. (We only consider moves of the spoiler in 
$\I(\AA, \S_{s}, s)$.
Moves in $\I(\BB, \T_{t}, t)$ are treated analogously.)
Any point or set move which is  entirely in 
$\I(\AA, \S_{s'}, s')$ is answered according to
the winning strategy on the substructures
$\I(\AA, \S_{s'}, s')$ and $\I(\BB, \T_{t'}, t')$.
For moves involving $a_0$, we proceed as follows. If the duplicator picks
$a_0$ in a point move, then the duplicator answers with 
$b_0$. Likewise, if the spoiler makes a set move
of the form $P \cup \{a_0\}$, where $P$ is a subset of the 
elements in $\I(\AA, \S_{s'}, s')$
then the duplicator answers with 
$Q \cup \{b_0\}$, where $Q$ is the duplicator's answer to $P$ in the game
played on the substructures 
$\I(\AA, \S_{s'}, s')$ and $\I(\BB, \T_{t'}, t')$.

Let $c_1, \dots, c_m$ and $d_1, \dots, d_m$ be the elements selected
in point moves and $P_1, \dots, P_n$ and $Q_1, \dots, Q_n$ be the
sets selected in set moves. By the above definition of the duplicator's
strategy, every move involving $a_0$ is answered by the analogous
move involving $b_0$. For all other elements, the selected elements clearly
define a partial isomorphism on the structures 
$\I(\AA, \S_{s'}, s')$ and $\I(\BB, \T_{t'}, t')$ extended by the selected sets.
It remains  to verify that the selected elements also define a partial isomorphism
on the structures $\I(\AA, \S_{s}, s)$ and $\I(\BB, \T_{t}, t)$
extended by the selected sets.
In particular, we have to verify that all relations $R \in \tau$
are preserved by the selected elements.
For any tuples of elements not involving $a_0$ (resp.\ $b_0$), this is
guaranteed by the fact that the winning strategy on
$\I(\AA, \S_{s'}, s')$ and $\I(\BB, \T_{t'}, t')$
is taken over to the structures 
$\I(\AA, \S_{s}, s)$ and $\I(\BB, \T_{t}, t)$.
On the other hand, 
by the connectedness condition of tree decompositions, 
we can be sure that the only relations on 
$\I(\AA, \S_{s}, s)$ (resp.\ $\I(\BB, \T_{t}, t)$)
involving 
$a_0$ (resp.\ $b_0$) are with elements in the bag
$(a_0, \dots, a_w)$ (resp.\ $(b_0, \dots, b_w)$). But then, by the 
equivalence $(a_0, \dots, a_w) \equiv (b_0, \dots, b_w)$, 
the preservation of $R \in \tau$ is again guaranteed.

\medskip

\noindent
(3) By the definition of branch nodes, the three nodes
$s, s_1, s_2$  have identical bags, say $(a_0, \dots, a_w)$. 
In particular, 
since the bag of $s$ introduces no new elements,
all elements contained in 
$\I(\AA, \S_{s}, s)$
are either contained in 
$\I(\AA, \S_{s_1}, s_1)$
or in 
$\I(\AA, \S_{s_2}, s_2)$.
Moreover, by the connectedness condition, only the elements
$a_0, \dots, a_w$ occur in both substructures.
Of course, the analogous observation holds for $t, t_1, t_2$, and 
$\I(\BB, \T_{t}, t)$.

By assumption, 
$\I(\AA, \S_{s_1}, s_1) \mequiv 
\I(\BB, \T_{t_2}, t_2)$ and 
$\I(\AA, \S_{s_2}, s_2) \mequiv 
\I(\BB, \T_{t_2}, t_2)$.
We define the 
duplicator's strategy on 
$\I(\AA, \S_{s}, s)$ and
$\I(\BB, \T_{t}, t)$ by 
simply combining the winning strategies on the substructures in the obvious way (Again we only consider moves of the spoiler in $\I(\AA, \S_{s}, s)$), i.e., 
if the spoiler picks  some element $c$ of 
$\I(\AA, \S_{s}, s)$ then the chosen element $c$ is in 
$\I(\AA, \S_{s_i}, s_i)$ for some $i \in \{ 1,2\}$. Hence,  the duplicator simply answers according to his winning strategy in the game on 
$\I(\AA, \S_{s_i}, s_i)$ and $\I(\BB, \T_{t_i}, t_i)$.
On the other hand, suppose that 
the spoiler picks a set $P$. Then $P$ is of the form $P = P_1 \cup P_2$, where $P_i$ contains only elements in $\I(\AA, \S_{s_i}, s_i)$. Thus, 
the duplicator simply answers with $Q = Q_1 \cup Q_2$, where $Q_i$ is the answer to $P_i$ according to the winning strategy in the game on 
$\I(\AA, \S_{s_i}, s_i)$ and $\I(\BB, \T_{t_i}, t_i)$.

It remains to verify that the selected vertices indeed define a partial isomorphism
on the structures 
$\I(\AA, \S_{s}, s)$ and
$\I(\BB, \T_{t}, t)$ extended by the selected sets. Again, the only 
interesting point is that every relation $R \in \tau$ is 
preserved by the elements selected in the point moves. If all elements in 
a tuple $\bar{c}$ (resp.\ $\bar{d}$) come from the same
substructure $\I(\AA, \S_{s_i}, s_i)$ (resp.\ $\I(\BB, \T_{t_i}, t_i)$), then
this is clearly fulfilled due to the fact that 
the duplicator's winning strategy on the substructures
$\I(\AA, \S_{s_i}, s_i)$ and $\I(\BB, \T_{t_i}, t_i)$ is taken over unchanged
to the game on $\I(\AA, \S_{s}, s)$ and $\I(\BB, \T_{t}, t)$. On the other hand,
by the connectedness condition, 
% of tree decompositions, 
we can be sure that the only relations between elements from different 
substructures $\I(\AA, \S_{s_1}, s_1)$ and $\I(\AA, \S_{s_2}, s_2)$ 
(resp.\ $\I(\BB, \T_{t_1}, t_1)$ and $\I(\BB, \T_{t_2}, t_2)$) are 
with elements in the bag
$(a_0, \dots, a_w)$ (resp.\ $(b_0, \dots, b_w)$) of $s_1$, $s_2$, and $s$
(resp.\ $t_1$, $t_2$, and $t$).
But then, by the 
equivalences
$\I(\AA, \S_{s_1}, s_1) \mequiv 
\I(\BB, \T_{t_1}, t_1)$ and 
$\I(\AA, \S_{s_2}, s_2) \mequiv 
\I(\BB, \T_{t_2}, t_2)$, 
the preservation of $R \in \tau$ is again guaranteed.
$\punto$
\end{proof}

\begin{lemma}
% [top-down traversal of the tree decomposition]
\label{lem:folklore-topDown}
Let $\A$ and $\B$ be $\tau$-structures, let
$\S$ (resp.\ $\T$) be a normalized tree decomposition of $\A$ (resp.\ of $\B$)
of width $w$, and let $s$ (resp.\ $t$) be an internal node in $\S$ (resp.\ in $\T$).

\smallskip

\noindent
(1) {\em permutation nodes.} Let $s'$ (resp.\ $t'$) be the only child of $s$ in $\S$ (resp.\ of $t$ in $\T$). Moreover, let 
% $(a_0, \dots, a_w)$, $(a'_0, \dots, a'_w)$, 
% $(b_0, \dots, b_w)$, and $(b'_0, \dots, b'_w)$
$\bar{a}$, $\bar{a}'$, $\bar{b}$, and $\bar{b}'$ 
denote the bags at the nodes $s$, $s'$, $t$, and $t'$, respectively.

\smallskip

\noindent
If $\I(\AA,\bar{\S}_{s},s) \mequiv \I(\BB, \bar{\T}_{t},t)$ 
and there exists a permutation
$\pi$, s.t.
$\bar{a} = \pi(\bar{a}')$ and
$\bar{b} = \pi(\bar{b}')$ 

then
$\I(\AA, \bar{\S}_{s'},s') \mequiv \I(\BB, \bar{\T}_{t'},t')$.

\medskip

\noindent
(2) {\em element replacement nodes.} Let $s'$ (resp.\ $t'$) be the 
only child  of $s$ in $\S$ (resp.\ of $t$ in $\T$). Moreover, let 
$\bar{a} = (a_0, a_1$, $\dots, a_w)$, 
$\bar{a}' = (a'_0, a_1, \dots, a_w)$, 
$\bar{b} = (b_0, b_1,$ $\dots, b_w)$, 
and $\bar{b}' = (b'_0, b_1, \dots, b_w)$
denote the bags at the nodes $s$, $s'$, $t$, and $t'$, respectively.

\smallskip

\noindent
If $\I(\AA,\bar{\S}_{s},s) \mequiv 
\I(\BB,\bar{\T}_{t},t) $ and 
$\bar{a}'  \equiv \bar{b}'$ then
$\I(\AA,\bar{\S}_{s'},s')$ $\mequiv 
\I(\BB,\bar{\T}_{t'},t') $.

\medskip

\noindent
(3) {\em branch nodes.} Let $s_1$ and $s_2$  (resp.\ $t_1$ and $t_2$) be the children of $s$ in $\S$ (resp.\ of $t$ in $\T$). 

\smallskip

\noindent
If $\I(\AA,\bar{\S}_{s},s) \mequiv \I(\BB, \bar{\T}_{t},t)$ and
$\I(\AA,\S_{s_2},s_2) \mequiv \I(\BB$, $\T_{t_2},t_2)$  then

$\I(\AA,\bar{\S}_{s_1},s_1) \mequiv \I(\BB, \bar{\T}_{t_1},t_1)$. 

\smallskip

\noindent
If $\I(\AA,\bar{\S}_{s},s) \mequiv \I(\BB, \bar{\T}_{t},t)$ and
 $\I(\AA,\S_{s_1},s_1) \mequiv \I(\BB$, $T_{t_1},t_1)$  then

 $\I(\AA,\bar{\S}_{s_2},s_2) \mequiv \I(\BB, \bar{\T}_{t_2},t_2)$. 
\end{lemma}

\begin{proof} The proof is by 
Ehrenfeucht-Fra\"{\i}ss\'{e} games,  
analogously to 
Lemma~\ref{lem:folklore-bottomUp}
$\punto$
\end{proof}

\begin{lemma}
% [combining bottom-up and top-down]
\label{lem:gluing}
Let $\A$ and $\B$ be $\tau$-structures, let
$\S$ (resp.\ $\T$) be a normalized tree decomposition of $\A$ (resp.\ of $\B$)
of width $w$, and let $s$ (resp.\ $t$) be an arbitrary node in $\S$ (resp.\ in $\T$), whose bag is $(a_0, \dots, a_w)$ (resp.\ $(b_0, \dots, b_w)$).

\smallskip

\noindent
If $\I(\AA,\S_{s},s) \mequiv \I(\BB, \T_{t},t)$ and
$\I(\AA,\bar{\S}_{s},s) \mequiv \I(\BB$, $\bar{\T}_{t},t)$
then
$({\cal A}, a_i) \mequiv ({\cal B}, b_i)$ for every $i \in \{0, \dots, w\}$.

\end{lemma}

\begin{proof} Again, the proof is by 
Ehrenfeucht-Fra\"{\i}ss\'{e} games,  
analogously to 
Lemma~\ref{lem:folklore-bottomUp}
$\punto$
\end{proof}

\noindent
{\em Discussion.}
Lemma~\ref{lem:folklore-bottomUp} provides the intuition
how to determine the $k$-type of the substructure induced by a 
subtree $\S_{s}$ via a bottom-up traversal of the tree decomposition $\S$. 
The three cases in the lemma refer to the three kinds of nodes which the 
root node $s$ of this subtree can have. The essence of the lemma is 
that the type of the structure induced by $\S_{s}$ is 
fully determined by the type of the structure induced by the
subtree rooted at the child node(s) plus the relations between elements in
the bag at node $s$. Of course, this is no big surprise.
Analogously, Lemma~\ref{lem:folklore-topDown} deals with the
$k$-type of the substructure induced by a 
subtree $\bar{\S}_{s}$, which can be obtained via a top-down traversal of 
$\S$. Finally, Lemma~\ref{lem:gluing} shows how
the $k$-type of the substructures induced by $\S_{s}$ and $\bar{\S}_{s}$
fully determines the type of the entire structure $\AA$ extended by some
domain element from the bag of $s$.

\section{Monadic Datalog}
\label{sec:monDatalog}

In this section, we introduce two restricted fragments of datalog,
namely  {\em monadic datalog\/} over finite 
structures with bounded treewidth and the {\em quasi-guarded fragment}
of datalog. 
Let $\tau = \{R_1, \dots, R_K\}$ be a set of predicate symbols and let $w \geq 1$ denote
the treewidth. Then we define the following extended signature
$\tauTD$. 

%\smallskip

\begin{center}
$
\tauTD = \tau \cup \{ \root, \leaf, 
\child_1, \child_2, \bag \}
$
\end{center}

%\smallskip
\noindent
where the unary predicates $\root$, and $\leaf$ as well
as the binary predicates $\child_1$ and $\child_2$ are used
to represent the tree  $T$ of the normalized 
tree decomposition in the obvious way. 
For instance, we write 
$\child_1(s_1,s)$ to denote that $s_1$ is either the first child or the only child
of $s$. 
Finally, $\bag$ has arity $w+2$, where $\bag(t, a_0, \dots, a_w)$
means that the bag at node $t$ is $(a_0, \dots, a_w)$.

\begin{definition}
\label{def:MonDatalog}
Let $\tau$ be a set of predicate symbols and let
$w \geq 1$.
A monadic datalog program over $\tau$-structures with 
treewidth $w$ is a set of datalog rules where all
extensional predicates are from $\tauTD$ and
all intensional predicates are unary.
\end{definition}
For any $\tau$-structure $\AA$ with normalized tree decomposition 
$\TT = \langle T, (A_t)_{t\in T} \rangle$ 
of width $w$, we denote by $\AATD$
the $\tauTD$-structure representing $\AA$ plus 
$\TT $ 
as follows:
The domain of $\AATD$ is the union of $\domA$ and the nodes
of $T$. 
In addition to the relations $\RAi$ with 
$R_i \in \tau$, the structure
$\AATD$ also contains relations for each predicate
$\root$, $\leaf$, 
$\child_1$, $\child_2$, and $\bag$ thus representing
the tree decomposition $\TT$. 
By \cite{Bod96}, one can compute $\AATD$ from $\AA$ 
in linear time w.r.t.\
the size of $\AA$. 
Hence, the size of $\AATD$ 
(for some reasonable encoding, see e.g.\ 
\cite{flum-frick-grohe-JACM02})
is also 
linearly bounded by the size of $\AA$.

\begin{example}
\label{bsp:tauTD-structure}

Recall the 
relational schema $(R,F)$ represented by the structure $\AA$ from Example \ref{bsp:tree-decomposition}
with normalized tree decomposition $\TT'$ in
Figure~\ref{fig:normalized-tree-decomp}. 
The domain of $\AATD$ is the union of $\domA$ and the tree nodes $\{s_1, \ldots, s_{22}\}$.
The corresponding $\tauTD$ structure $\AATD$ representing
the relational schema plus tree decomposition $\TT'$
is made up by the following set of ground atoms:
$\root(s_1)$, $\leaf(s_{12})$, $\leaf(s_{14})$, $\leaf(s_{19})$,
$\child_1(s_2, s_1)$, $\child_2(s_3, s_1)$, $\ldots$,
$\bag(s_1, f_3, d,e)$, $\ldots$.
\nop{$\AATD$ als set of ground atoms fuer dieses Beispiel angeben.}
\end{example}

As we recalled in Section~\ref{sec:Datalog}, the
evaluation of monadic datalog is NP-complete
(combined complexity). However, the target of our transformation
from MSO to datalog will be a further restricted fragment of datalog,
which we refer to as ``quasi-guarded''. The evaluation of
this fragment can be easily shown to be tractable.

\begin{definition}
\label{def:QuasiGuarded}
Let $B$ be an atom and $y$ a variable in some rule $r$.
 We call
$y$  ``functionally dependent'' on $B$ if in every ground 
instantiation $r'$ of $r$, the value of $y$ is uniquely determined 
by the value of $B$.

We call a datalog program $\P$ ``quasi-guarded'' if
every rule $r$ contains an extensional atom $B$,
s.t.\ every variable occurring in $r$ either occurs in $B$ or
is functionally dependent on  $B$.
\end{definition}

\begin{theorem}
\label{theo:mondatalog-complexity}
Let $\P$ be a quasi-guarded datalog program and
let $\AA$ be a finite structure. Then $\P$ can be 
evaluated over $\AA$ in time 
${\cal O}(|\PP| * |\AA|)$, where
$|\PP|$ denotes the size of the datalog
program and $|\AA|$ denotes the size 
of the data.
\end{theorem}

\begin{proof}
Let $r$ be a rule in the program $\P$ and let 
$B$ be the ``quasi-guard'' of $r$, i.e., all variables
in $r$ either occur in $B$ or are functionally dependent on $B$.
In order to compute all possible ground instances
$r'$ of $r$ over $\AA$, we first instantiate $B$. The maximal number of 
such instantiations is clearly bounded by $|\AA|$.
Since all other variables occurring in $r$ are functionally dependent on the 
variables in $B$, in fact the number of 
all possible ground instantiations $r'$ of $r$ is bounded by $|\AA|$. 

Hence, in total, the ground program $\P'$ consisting of all possible
ground instantiations of the rules in $\P$ has size 
${\cal O}(|\PP| * |\AA|)$ and also the computation of these ground
rules fits into the linear time bound. As we 
recalled in Section~\ref{sec:Datalog}, the ground program $\P'$ 
can be evaluated over $\AA$ in time 
${\cal O}(|\PP'| + |\AA|) = {\cal O}((|\PP| * |\AA|) + |\AA|)
= {\cal O}(|\PP| * |\AA|)$.
$\punto$
\end{proof}

Before we state the main result concerning the 
{\em expressive power\/} of monadic datalog over structures with
bounded treewidth,
we introduce the following notation. In order to simplify the exposition
below, we assume that all predicates $R_i \in \tau$ have the same arity $r$.
First, this can be easily achieved by copying columns in relations
with smaller arity. Moreover, it is easily seen that the results also hold
without this restriction. 

It is convenient to use the following abbreviations. Let 
$\bar{a} = (a_0, \dots, a_w)$ be a tuple of domain elements. Then 
we write $\R(\bar{a})$ to denote the set of all ground atoms
with predicates in $\tau = \{R_1, \dots, R_K\}$ 
and arguments in $\{a_0, \dots, a_w\}$, i.e.,
$$\R(\bar{a}) = 
\bigcup_{i = 1}^{K} \  \
\bigcup_{j_1 = 0}^w \dots
\bigcup_{j_r = 0}^w 
\{R_i(a_{j_1}, \ldots, a_{j_r})\}
$$
%
%  We write $\R_0(\bar{a})$ for the set of those ground atoms
%  in $\R(\bar{a})$ where the element $a_0$ occurs as 
%  an argument, i.e.,
%
%  $$
%  \R_0(\bar{a}) = \bigcup^{i \in \{1..k\}}_{i_j 
%  \in \{0..w\}_{(1 \leq j \leq r)} \wedge 0 
%  \in \{i_1, \ldots, i_r\}} 
%  \{R_i(a_{i_1}, \ldots, a_{i_r})\}
%  $$
%
Let $\AA$ be a structure  with tree decomposition $\TT$ and
let $s$ be a node in $\TT$ whose bag is $\bar{a} = (a_0, \dots, a_w)$.
Then we write $(\AA,s)$ as a short-hand for the 
structure $(\AA,\bar{a})$ with distinguished
constants $\bar{a} = (a_0, \dots, a_w)$.

\begin{theorem}
\label{theo:MSO-to-Datalog}
Let $\tau$ and $w \geq 1$ be arbitrary but fixed.
Every MSO-definable unary query over $\tau$-structures of 
tree\-width $w$ is 
also 
definable in the quasi-guarded fragment of 
monadic datalog over $\tauTD$.
\end{theorem}

\begin{proof}
Let $\phi(x)$ be an arbitrary MSO formula with free variable $x$ and 
quantifier depth $k$. We have to construct a monadic datalog program $\P$
with distinguished predicate $\phi$ which defines the same query.

W.l.o.g., we only consider the case of structures whose domain has
$\geq w +1$ elements.
We maintain two disjoint sets of $k$-types $\ThetaUp$ and $\ThetaD$, representing
$k$-types of structures $(\AA,\bar{a})$ of the following form:
$\AA$ has a tree decomposition $\TT$ of width $w$ and $\bar{a}$
is the bag of some node $s$ in $\TT$. Moreover, for $\ThetaUp$,
we require that $s$ is the root of $\SS$ while, for
$\ThetaD$, we require that $s$ is a leaf node of $\TT$. 
We maintain for each
type $\theta$ a witness $W(\theta) = \ATs$.
The types in $\ThetaUp$ and $\ThetaD$ will serve as predicate names
in the monadic datalog program to be constructed.
Initially, $\ThetaUp = \ThetaD = \P = \emptyset$.

%%%%%%%%%%%%%%%%%%%%%%%%%%%%%%%%%%%
% Bottom up
%%%%%%%%%%%%%%%%%%%%%%%%%%%%%%%%%%%

\medskip

\noindent
1. {\em ``Bottom-up'' construction of $\ThetaUp$.}

\noindent
{\sc Base Case.} Let $a_0, \ldots, a_w$ be pairwise
distinct elements and let $\SS$ be a tree decomposition
consisting of a single node $s$, whose bag is
$A_s = (a_0, \ldots, a_w)$. Then we consider
all possible structures $(\AA, s)$ with this 
tree decomposition. In particular, 
$\dom(\AA) = \{a_0, \dots, a_w\}$. We
get all possible structures with tree decomposition $\SS$
by letting the EDB $\EE(\AA)$ be any 
subset of $\R(\bar{a})$.
For every such structure $(\AA, s)$, we 
check if there exists a type $\theta \in \ThetaUp$
with $ W(\theta) = \BTt$, s.t.\ $(\AA,s) \mequiv (\BB,t)$.
If such a $\theta$ exists, we take it. 
Otherwise we invent a new
token $\theta$, add it to $\ThetaUp$ and set 
$W(\theta):= \ASs$. In any case, we add the following
rule to the program $\P$:

\medskip

\noindent
\[
\begin{array}{lll}
\theta (v) & \la & \bag(v,x_0, \dots, x_w), \leaf(v), 
\{R_i (x_{j_1}, \dots, x_{j_r}) \mid R(a_{j_1}, \dots, a_{j_r}) \in \EE(\AA) \}, \\
&&\{\neg R_i (x_{j_1}, \dots, x_{j_r}) \mid R(a_{j_1}, \dots, a_{j_r}) \not\in \EE(\AA) \}. 
\end{array}
\]

\medskip

\noindent
{\sc Induction step.}
We construct new structures by extending the tree decompositions
of existing witnesses in ``bottom-up'' direction, i.e., by introducing
a new root node. This root node may be one of three kinds of nodes. 
% in a normalized tree decomposition.

\smallskip

\noindent
(a) Permutation nodes.
For each $\theta' \in \ThetaUp$, let
$W(\theta') = \ASsp$ with bag
$A_{s'} = (a_0, \ldots, a_w)$ at the root $s'$ in $\SS'$. 
Then we consider all possible triples $\ASs$, where
$\SS$ is obtained from $\SS'$ by appending $s'$ to a 
new root node $s$, s.t.\ $s$ is a permutation node, 
i.e., there exists some permutation $\pi$, s.t.\
$A_{s} = (a_{\pi(0)}, \ldots, a_{\pi(w)})$

For every such structure $(\AA, s)$, we 
check if there exists a type $\theta \in \ThetaUp$
with $ W(\theta) = \BTt$, s.t.\ $(\AA,s) \mequiv (\BB,t)$.
If such a $\theta$ exists, we take it. 
Otherwise we invent a new
token $\theta$, add it to $\ThetaUp$ and set 
$W(\theta):= \ASs$. In any case, we add the following
rule to the program $\P$:

\smallskip

\noindent

\[
\begin{array}{lll}
\theta (v)  & \la & \bag(v,x_{\pi(0)}, \ldots, x_{\pi(w)}), 
\child_1(v',v), 
\theta' (v'), \bag(v',x_0, \dots, x_w).
\end{array}
\]

\medskip

\noindent
(b) Element replacement nodes.
For each $\theta' \in \ThetaUp$, let
$W(\theta') = \ASspp$ with bag
$A_{s'} = (a'_0, a_1, \ldots, a_w)$ at the root $s'$ in $\SS'$. 
% Let $A'$ denote the domain of $\AA'$.
Then we consider all possible triples $\ASs$, where
$\SS$ is obtained from $\SS'$ by appending $s'$ to a 
new root node $s$, s.t.\ $s$ is an element 
replacement node. For the tree decomposition $\SS$, we thus
 invent some new element $a_0$ 
and set $A_{s} = (a_0, a_1, \ldots, a_w)$. 
For this
tree decomposition $\SS$, we 
consider all possible structures $\AA$ with 
$\dom (\AA) = \dom (A') \cup \{a_0\}$
where the EDB $\EE(\AA')$ is extended to the EDB
$\EE(\AA)$ by new ground atoms from $\R(\bar{a})$, s.t.\
$a_0$ occurs as argument of 
all ground atoms in
$\EE(\AA) \setminus \EE(\AA')$.

For every such structure $(\AA, s)$, we 
check if there exists a type $\theta \in \ThetaUp$
with $ W(\theta) = \BTt$, s.t.\ $(\AA,s) \mequiv (\BB,t)$.
If such a $\theta$ exists, we take it. 
Otherwise we invent a new
token $\theta$, add it to $\ThetaUp$ and set 
$W(\theta):= \ASs$. In any case, we add the following
rule to the program $\P$:

\smallskip

\noindent

\[
\begin{array}{lll}
\theta (v) & \la & \bag(v,x_0, x_1, \dots, x_w),
\child_1(v',v),  \theta' (v'), \bag(v', x_0', x_1, \dots, x_w), \\
& &\{R_i (x_{j_1}, \dots, x_{j_r}) \mid R(a_{j_1}, \dots, a_{j_r}) \in \EE(\AA) \}, \\
& & \{\neg R_i (x_{j_1}, \dots, x_{j_r}) \mid R(a_{j_1}, \dots, a_{j_r}) \not\in \EE(\AA) \}. 
\end{array}
\]

\medskip

\noindent
(c) Branch nodes.
Let $\theta_1$, $\theta_2$ be two (not necessarily distinct) 
types in $\ThetaUp$ 
with $W(\theta_1) = 
\langle {\cal A}_1, {\cal S}_1, s_1 \rangle$
and 
$W(\theta_2) = 
\langle {\cal A}_2, {\cal S}_2, s_2 \rangle$. Let
$A_{s_1} = (a_0, \ldots, a_w)$ and
$A_{s_2} = (b_0, \ldots, b_w)$, respectively.
Moreover, let $\dom(\AA_1) \cap \dom(\AA_2) = \emptyset.$

Let $\delta$ be a renaming function with
$\delta = \{a_0 \leftarrow b_0, \ldots, a_w \leftarrow b_w\}$.
By applying $\delta$ to $\langle {\cal A}_2, {\cal S}_2, s_2 \rangle$,
we obtain a new triple
$\langle {\cal A}'_2, {\cal S}'_2, s_2 \rangle$ with
${\cal A}'_2 = {\cal A}_2 \delta$ and
${\cal S}'_2 = {\cal S}_2 \delta$. In particular, we
thus have $A_{s_2}\delta = (a_0, \ldots, a_w)$.
Clearly, 
$( {\cal A}_2, s_2) 
\mequiv ( {\cal A}'_2, s_2)$ holds.

For every such pair $\langle {\cal A}_1, {\cal S}_1, s_1 \rangle$ 
and $\langle {\cal A}'_2, {\cal S}'_2, s_2 \rangle$,
we check if the EDBs are inconsistent, i.e.,
$\EE(\AA_1) \cap \R(\bar{a}) \neq \EE(\AA'_2) \cap \R(\bar{a})$.
If this is the case, then we ignore this pair. Otherwise, 
we construct a new tree decomposition $\SS$ with a 
new root node $s$, whose child nodes are $s_1$ and $s_2$.
As the bag of $s$, we set $A_s = A_{s_1} = A_{s'_2}$. By construction,
$\SS$ is a normalized tree decomposition of the structure $\AA$
with $\dom(\AA) = \dom(\AA_1) \cup \dom(\AA'_2)$ and EDB 
$\EE(\AA) = \EE(\AA_1) \cup \EE(\AA'_2)$.

As in the cases above, we have to check
if there exists a type $\theta \in \ThetaUp$
with $ W(\theta) = \BTt$, s.t.\ $(\AA,s) \mequiv (\BB,t)$.
If such a $\theta$ exists, we take it. 
Otherwise we invent a new
token $\theta$, add it to $\ThetaUp$ and set 
$W(\theta):= \ASs$. In any case, we add the following
rule to the program $\P$:

\smallskip

\noindent

\[
\begin{array}{lll}
\theta (v)  & \la &  \bag(v,x_0, x_1, \dots, x_w),
\child_1(v_1,v),\theta_1 (v_1), \child_2(v_2,v),\theta_2 (v_2), \\
& & \bag(v_1,x_0, x_1, \dots, x_w), \bag(v_2,x_0, x_1, \dots, x_w).
\end{array}
\]

%%%%%%%%%%%%%%%%%%%%%%%%%%%%%%%%%%%
% Top down
%%%%%%%%%%%%%%%%%%%%%%%%%%%%%%%%%%%
\medskip

\noindent
2. {\em ``Top-down'' construction of $\ThetaD$.} 

\nop{****************
\smallskip
\noindent
Analogously to the ``bottom-up'' construction of $\ThetaUp$, 
we construct the set $\ThetaD$ of types with a ``top-down'' 
intuition. The base case is essentially the same as before since, 
in every tree decomposition with only one node $s$, this single
node is both the root and a leaf. For the induction step, 
we have to select the witness $W(\theta') = \ASspp$ of 
some already computed  type $\theta' \in \ThetaD$. Now  
the node $s'$ in $\SS'$ is a leaf node and we extend 
% the tree decomposition 
$\SS'$ to a new tree decomposition $\SS$
by appending a new leaf node $s$ as a child of 
$s'$. For all such tree decompositions
$\SS$, we consider all possible structures
$\AA$ by appropriately extending $\AA'$. The rules
added to the program $\P$ again reflect the 
type transitions from the type of the original structure
$(\AA',s')$ to the type of any such new structure $(\AA,s)$.
****************}

\noindent
{\sc Base Case.} Let $a_0, \ldots, a_w$ be pairwise
distinct elements and let $\SS$ be a tree decomposition
consisting of a single node $s$, whose bag is
$A_s = (a_0, \ldots, a_w)$. Then we consider
all possible structures $(\AA, s)$ with this 
tree decomposition. In particular, 
$\dom(\AA) =  \{a_0, \dots, a_w\}$. 
We get all possible structures with tree decomposition $\SS$
by letting the EDB $\EE(\AA)$ be any 
subset of $\R(\bar{a})$.
For every such structure $(\AA, s)$, we 
check if there exists a type $\theta \in \ThetaD$
with $ W(\theta) = \BTt$, s.t.\ $(\AA,s) \mequiv (\BB,t)$.
If such a $\theta$ exists, we take it. 
Otherwise we invent a new
token $\theta$, add it to $\ThetaD$ and set 
$W(\theta):= \ASs$. In any case, we add the following
rule to the program $\P$:

\smallskip

\noindent

\[
\begin{array}{lll}
\theta (v) & \la & \bag(v,x_0, \dots, x_w), \root(v),
\{R_i (x_{j_1}, \dots, x_{j_r}) \mid R(a_{j_1}, \dots, a_{j_r}) \in \EE(\AA) \}, \\
& & \{\neg R_i (x_{j_1}, \dots, x_{j_r}) \mid R(a_{j_1}, \dots, a_{j_r}) \not\in \EE(\AA) \}. 
\end{array}
\]

\medskip

\noindent
{\sc Induction step.}
We construct new structures by extending the tree decompositions
of existing witnesses in ``top-down'' direction, i.e., by introducing
a new leaf node $s$ and appending it as new child to a former 
leaf node $s'$. The node $s'$ may thus become one of three kinds of 
nodes in a normalized tree decomposition.

\smallskip

\noindent
(a) Permutation nodes.
For each $\theta' \in \ThetaD$, let
$W(\theta') = \ASsp$ with bag
$A_{s'} = (a_0, \ldots, a_w)$ at some leaf node $s'$ in $\SS'$. 
Then we consider all possible triples $\ASs$, where
$\SS$ is obtained from $\SS'$ by appending $s$ as a
new child of $s'$, s.t.\ $s'$ is a permutation node, 
i.e., there exists some permutation $\pi$, s.t.\
$A_{s} = (a_{\pi(0)}, \ldots, a_{\pi(w)})$

For every such structure $(\AA, s)$, we 
check if there exists a type $\theta \in \ThetaD$
with $W(\theta) = \BTt$, s.t.\ $(\AA,s) \mequiv (\BB,t)$.
If such a $\theta$ exists, we take it. 
Otherwise we invent a new
token $\theta$, add it to $\ThetaD$ and set 
$W(\theta):= \ASs$. In any case, we add the following
rule to the program $\P$:

\medskip

\noindent
\[
\begin{array}{lll}
\theta (v)  & \la & \bag(v,x_{\pi(0)}, \ldots, x_{\pi(w)}), \child_1(v, v'),
\theta' (v'), \bag(v',x_0, \dots, x_w).
\end{array}
\]

\medskip

\noindent
(b) Element replacement nodes.
For each $\theta' \in \ThetaD$, let
$W(\theta') = \ASspp$ with bag
$A_{s'} = (a'_0, a_1, \ldots$, $a_w)$ at leaf node
$s'$ in $\SS'$. 
Then we consider all possible triples $\ASs$, where
$\SS$ is obtained from $\SS'$ by appending $s$ 
as new child of $s'$, s.t.\ $s'$ is an element 
replacement node. For the tree decomposition $\SS$, we thus
invent some new element $a_0$ 
and set $A_{s} = (a_0, a_1, \ldots, a_w)$. For this
tree decomposition $\SS$, we 
consider all possible structures $\AA$ with 
$\dom(\AA) = \dom(\AA') \cup \{a_0\}$
where the EDB $\EE(\AA')$ is extended to the EDB
$\EE(\AA)$ by new ground atoms 
from $\R(\bar{a})$, s.t.\
$a_0$ occurs as argument of 
all ground atoms in
$\EE(\AA) \setminus \EE(\AA')$.

For every such structure $(\AA, s)$, we 
check if there exists a type $\theta \in \ThetaD$
with $ W(\theta) = \BTt$, s.t.\ $(\AA,s) \mequiv (\BB,t)$.
If such a $\theta$ exists, we take it. 
Otherwise we invent a new
token $\theta$, add it to $\ThetaD$ and set 
$W(\theta):= \ASs$. In any case, we add the following
rule to the program $\P$:

\smallskip

\noindent

\[
\begin{array}{lll}
\theta (v) & \la & \bag(v,x_0, x_1, \dots, x_w),
\child_1(v,v'),  \theta' (v'), \bag(v', x_0', x_1, \dots, x_w), \\
& & \{R_i (x_{j_1}, \dots, x_{j_r}) \mid R(a_{j_1}, \dots, a_{j_r}) \in \EE(\AA) \}, \\
& & \{\neg R_i (x_{j_1}, \dots, x_{j_r}) \mid R(a_{j_1}, \dots, a_{j_r}) \not\in \EE(\AA) \}. 
\end{array}
\]

\medskip

\noindent
(c) Branch nodes.
Let $\theta \in \ThetaD$ and $\theta_2 \in \ThetaUp$ 
with $W(\theta) = 
\langle {\cal A}, {\cal S}, s \rangle$
and 
$W(\theta_2) = 
\langle {\cal A}_2, {\cal S}_2, s_2 \rangle$. Note that
$s$ is a leaf in $\S$ while $s_2$ is the root of $\S_2$.
Now let
$A_{s} = (a_0, \ldots, a_w)$ and
$A_{s_2} = (b_0, \ldots, b_w)$, respectively,
and let $\dom(\AA) \cap \dom(\AA_2) = \emptyset$.

Let $\delta$ be a renaming function with
$\delta = \{a_0 \leftarrow b_0, \ldots, a_w \leftarrow b_w\}$.
By applying $\delta$ to $\langle {\cal A}_2, {\cal S}_2, s_2 \rangle$,
we obtain a new triple
$\langle {\cal A}'_2, {\cal S}'_2, s_2 \rangle$ with
${\cal A}'_2 = {\cal A}_2 \delta$ and
${\cal S}'_2 = {\cal S}_2 \delta$. In particular, we
thus have $A_{s_2}\delta = (a_0, \ldots, a_w)$.
Clearly, 
$( {\cal A}_2, s_2) 
\mequiv ( {\cal A}'_2, s_2)$ holds. 

For every such pair $\langle {\cal A}, {\cal S}, s \rangle$ 
and $\langle {\cal A}'_2, {\cal S}'_2, s_2 \rangle$,
we check if the EDBs are inconsistent, i.e.,
$\EE(\AA) \cap \R(\bar{a}) \neq \EE(\AA'_2) \cap \R(\bar{a})$.
If this is the case, then we ignore this pair. Otherwise, 
we construct a new tree decomposition $\SS_1$ by introducing a 
new leaf node $s_1$ and appending both $s_1$ and $s_2$ as child nodes
of $s$.
As the bag of $s_1$, we set $A_{s_1} = A_{s} = A_{s'_2}$. 
By construction,
$\SS_1$ is a normalized tree decomposition of the structure $\AA_1$
with $\dom(\AA_1) = \dom(\AA) \cup \dom(\AA'_2)$ and EDB 
$\EE(\AA_1) = \EE(\AA) \cup \EE(\AA'_2)$.

As in the cases above, we have to check
if there exists a type $\theta_1 \in \ThetaD$
with $ W(\theta_1) = \BTt$, s.t.\ $(\AA_1,s_1) \mequiv (\BB,t)$.
If such a $\theta_1$ exists, we take it. 
Otherwise we invent a new
token $\theta_1$, add it to $\ThetaD$ and set 
$W(\theta_1):= \langle {\cal A}_1, {\cal S}_1, s_1 \rangle$
In any case, we add the following
rule to the program $\P$:

\smallskip

\noindent

\[
\begin{array}{lll}
\theta_1 (v_1) & \la &  \bag(v_1,x_0, x_1, \dots, x_w),
\child_1(v_1,v), \child_2(v_2,v),\theta (v), \theta_2 (v_2), \\
& & \bag(v,x_0, x_1, \dots, x_w), \bag(v_2,x_0, x_1, \dots, x_w).
\end{array}
\]

\smallskip

\noindent
Now suppose that $\SS_1$ is constructed from $\SS$ and $\SS_2$ by 
attaching the new node $s_1$ as second child of $s$ and $s_2$ as the
first child. In this case, the structure $\AA_1$ remains exactly the same 
as in the case above, since the order of the child nodes of a 
node in the tree decomposition is irrelevant. Thus, whenever the
above rule is added to the program $\P$, then also the following rule
is added:

\smallskip

\noindent

\[
\begin{array}{lll}
\theta_1 (v_2) & \la & \bag(v_2,x_0, x_1, \dots, x_w),
\child_1(v_1,v), \child_2(v_2,v), \theta (v), \theta_2 (v_1), \\
& & \bag(v,x_0, x_1, \dots, x_w), \bag(v_1,x_0, x_1, \dots, x_w).
\end{array}
\]

\nop{****************
\smallskip
\noindent
The correctness of this construction is an immediate consequence 
of Lemma~\ref{lem:folklore-topDown}. 
****************} %nop

\medskip

\noindent
3. {\em Element selection.} 

\smallskip

\noindent
We consider all pairs of types $\theta_1 \in \ThetaUp$ and 
$\theta_2 \in \ThetaD$. 
Let $W(\theta_1) = 
\langle {\cal A}_1, {\cal S}_1, s_1 \rangle$
and 
$W(\theta_2) = 
\langle {\cal A}_2, {\cal S}_2, s_2 \rangle$. Moreover, 
let
$A_{s_1} = (a_0, \ldots, a_w)$ and
$A_{s_2} = (b_0, \ldots, b_w)$, respectively, 
and let $\dom(\AA_1) \cap \dom(\AA_2) = \emptyset$.

Let $\delta$ be a renaming function with
$\delta = \{a_0 \leftarrow b_0, \ldots, a_w \leftarrow b_w\}$.
By applying $\delta$ to $\langle {\cal A}_2, {\cal S}_2, s_2 \rangle$,
we obtain a new triple
$\langle {\cal A}'_2, {\cal S}'_2, s_2 \rangle$ with
${\cal A}'_2 = {\cal A}_2 \delta$ and
${\cal S}'_2 = {\cal S}_2 \delta$. In particular, we
thus have $A_{s_2}\delta = (a_0, \ldots, a_w)$.
Clearly, 
$( {\cal A}_2, s_2) 
\mequiv ( {\cal A}'_2, s_2)$ holds. 

For every such pair $\langle {\cal A}_1, {\cal S}_1, s_1 \rangle$ 
and $\langle {\cal A}'_2, {\cal S}'_2, s_2 \rangle$,
we check if the EDBs are inconsistent, i.e.,
$\EE(\AA_1) \cap \R(\bar{a}) \neq \EE(\AA'_2) \cap \R(\bar{a})$.
If this is the case, then we ignore this pair. Otherwise, 
we construct a new tree decomposition $\SS$ 
by identifying $s_1$ (= the root of $\SS_1$) with 
$s_2$ (= a leaf of $\SS_2$).
% The bag of $s$ is, of course,  $A_s = A_{s_1} = A_{s'_2}$. 
By construction,
$\SS$ is a normalized tree decomposition of the structure $\AA$
with $\dom(A) = \dom(\AA_1) \cup \dom(\AA'_2)$ and 
$\EE(\AA) = \EE(\AA_1) \cup \EE(\AA'_2)$.

Now check for each $a_i$ in $A_{s_1} = A_{s_2}\delta$, if 
$\AA \models \phi (a_i)$. If this is the case, then
we add the following rule to ${\cal P}$.
\smallskip

\noindent
\[
\begin{array}{lll}
\phi(x_i) & \la & \theta_1 (v), \theta_2 (v), bag(v,x_0, \ldots, x_w).
\end{array}
\]

\medskip

\noindent
We claim that the program $\P$ with distinguished monadic predicate 
$\phi$ is the desired monadic datalog program, i.e., let
$\AA$ be an arbitrary input $\tau$-structure with tree decomposition
$\SS$ and let $\AATD$ denote the corresponding  $\tauTD$-structure. Moreover,
let $a \in \dom(\AA)$. Then the following equivalence
holds:
$\AA \models \phi (a) \mbox{ iff } \phi(a) 
 \mbox{ is in the least fixpoint of } \P \cup \AATD.$

Note that the intensional predicates in $\ThetaUp$, $\ThetaD$,
and $\{\phi\}$ are layered in 
% the sense 
that we can first compute
the  least fixpoint of the predicates in $\ThetaUp$, then 
% of the predicates in 
$\ThetaD$, and finally $\phi$. 

The bottom-up construction of $\ThetaUp$ guarantees that we indeed
construct all possible types of structures $(\BB,t)$ with
tree decomposition $\TT$ and root $t$. This can be easily shown by 
Lemma~\ref{lem:folklore-bottomUp} and an induction on the size 
of the tree decomposition $\TT$. On the other hand, for 
every subtree $\SS_s$ of $\SS$, the
type of the induced substructure $\I(\AA,\SS_s,s)$ is $\theta$
for some $\theta \in \ThetaUp$
if and only if the atom $\theta(s)$ is in the 
 least fixpoint of 
$\P \cup \AATD$. Again this can be shown by an easy induction 
argument using Lemma~\ref{lem:folklore-bottomUp}.

Analogously, we may conclude via Lemma~\ref{lem:folklore-topDown} 
that $\ThetaD$ contains all possible types of structures $(\BB,t)$ with
tree decomposition $\TT$ and some leaf node $t$. Moreover, for 
every subtree $\bar{\SS}_s$ of $\SS$, the
type of the induced substructure $\I(\AA,\bar{\SS}_s,s)$ is $\theta$
for some $\theta \in \ThetaD$
if and only if the atom $\theta(s)$ is in the 
least fixpoint of 
$\P \cup \AATD$. 
The definition of the predicate $\phi$ in part 3 
is a direct realization of Lemma~\ref{lem:gluing}. 
It thus follows
that $\AA \models \phi(a)$ 
% for some domain element $a$ 
iff $\phi(a)$ is in the 
least fixpoint of $\P \cup \AATD$.

Finally, an inspection of all datalog rules
added to $\P$ by
this construction
shows that these rules are indeed quasi-guarded, i.e., they
all contain an atom $B$ with an extensional predicate, s.t.\
all other variables in this rule are functionally dependent
on the variables in $B$. For instance, in the rule added to 
$\ThetaUp$ in case of a branch node, the atom 
$bag(v,x_0, \ldots, x_w)$ is the quasi-guard. Indeed, the remaining
variables $v_1$ and $v_2$ in this rule are functionally dependent
on $v$ via the atoms $\child_1(v_1,v)$ and $\child_2(v_2,v)$.
% , respectively.
$\punto$
\end{proof}

\medskip

Above all, Theorem~\ref{theo:MSO-to-Datalog} is
an expressivity result. However, it can of course
be used to derive also a complexity result.
Indeed, we can state a slightly extended version of
Courcelle's Theorem  
as a corollary (which is in turn a special
case of Theorem 4.12 in \cite{flum-frick-grohe-JACM02}).

\begin{corollary}
The evaluation 
problem of unary MSO-queries $\phi(x)$ over
$\tau$-structures $\AA$ with treewidth $w$ 
can be solved
in
time ${\cal O}(f(|\phi(x)|, w) * |\A|)$ for some
function $f$.
\end{corollary}

\begin{proof}
Suppose that we are given 
an MSO-query $\phi(x)$ and some treewidth $w$.
By Theorem~~\ref{theo:MSO-to-Datalog}, 
we can construct an equivalent, quasi-guarded datalog
program $\P$. The whole construction is independent
of the data. Hence, the time for this construction and the 
size of $\P$ are both bounded by some term $f(|\phi(x)|, w)$.
By \cite{Bod96}, a tree decomposition $\TT$ of $\AA$ and, therefore, 
also the extended structure $\AATD$ can be computed in 
time ${\cal O}(|\A|)$. Finally, 
by Theorem~\ref{theo:mondatalog-complexity}, the 
quasi-guarded program $\P$ can be evaluated
over $\AATD$ in time ${\cal O}(|\P| * |\AATD|)$, from 
which the desired overall time bound follows. 
$\punto$
\end{proof}

\noindent
{\em Discussion.\/} Clearly, Theorem~\ref{theo:MSO-to-Datalog} is 
not only applicable to MSO-definable {\em unary queries\/} but also to 
{\em $0$-ary queries\/}, i.e., MSO-queries defining a decision
problem. An inspection of the proof of Theorem~\ref{theo:MSO-to-Datalog} reveals that several simplifications
are possible in this case. Above all, the whole
``top-down'' construction of $\ThetaD$ can be omitted. 
Moreover,
the rules with head predicate $\phi$ are now much simpler: Let 
$\phi$ be a $0$-ary MSO-formula and let 
$\ThetaUp$ denote the set of types obtained by the 
``bottom-up'' construction in the above proof. Then we define
$\ThetaUp_0 = \{ \theta \mid W(\theta) = \ASs$ and 
$\AA \models \phi\}$. 
%Let $\ThetaUp_0 = \{\theta_{i_1}, \dots, \theta_{i_m} \}$. 
Finally, we add the following set of rules with head predicate $\phi$
to our datalog
program:
\[
\begin{array}{lll}
\phi  & \la & \root(v), \theta_0(v).  
\end{array}
\]
for every $\theta_0 \in \ThetaUp_0$. We shall make use of these
simplifications in Section \ref{sec:three-col} and
\ref{sec:primality} when we present new algorithms for 
two decision problems. In contrast, these simplifications are no longer possible when we consider an enumeration problem in 
Section~\ref{sec:monadic-prime}. In particular, the ``top-down'' construction
will indeed be required then.

%%%%%%%%%%%%%%%%%%%%%%%%%%%%%%%%%%%%%%%%%%%%%%%%%%%%%%%%%%%%%%
%%%%%%%%%%%%%% SAT %%%%%%%%%%%%%%%%%%%%%%%%%%%%%%%%%%
\section{Monadic Datalog at Work}
\label{sec:Work}

We now put monadic datalog to work
by constructing several new algorithms. We start off with
a simple example, namely the 3-Colorability problem, which will help to illustrate the basic ideas, see Section \ref{sec:three-col}. Our ultimate goal is to 
tackle two more  involved problems, namely the 
PRIMALITY decision problem and the 
PRIMALITY enumeration problem,
see Sections \ref{sec:primality} and 
\ref{sec:monadic-prime}. All these problems 
are well-known to be intractable. However, 
since they are expressible in 
MSO over appropriate structures, they are
fixed-parameter tractable w.r.t.\
the treewidth. In this section, we show that these
problems admit succinct and efficient
solutions via datalog.

Before we present our datalog programs,
we slightly modify the notion of normalized tree decompositions from 
Section~\ref{sec:treewidth}.
Recall that an {\em element replacement node} 
replaces exactly one element in the bag of the 
child node by a new element. 
For our algorithms, it is preferable to split this action into two steps, namely, an 
{\em element removal node}, which removes one domain element
from the bag of its child node, and an 
{\em element introduction node}, which introduces one new element. 
Moreover, it is now preferable
to consider the bags as sets
of domain elements rather than as tuples. Hence,
we may delete permutation nodes from the tree decomposition.
Finally, we drop the condition that all bags in a tree decomposition
of width $w$  must have ``full size'' $w+1$
(by splitting the element replacement into 
element removal and element introduction, this condition would have 
required some relaxation anyway). Such a normal form
of tree decompositions was also considered in \cite{Kloks94}.
For instance, recall the tree decomposition ${\cal T}'$ from 
Figure~\ref{fig:normalized-tree-decomp}. A tree decomposition 
${\cal T}''$ compliant with our modified notion of {\em 
normalized tree decompositions\/} is depicted in 
Figure~\ref{fig:normalized-tree-decomp-no-perm}.

\begin{figure}[t]
\begin{center}
\includegraphics[scale=0.4]{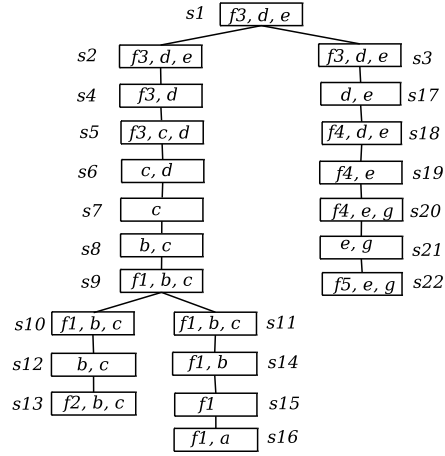}
\end{center}
\vspace{-0.5cm}
\caption{Modified normal form of tree decompositions.}
\label{fig:normalized-tree-decomp-no-perm}
\end{figure}

\subsection{The 3-Colorability Problem}
\label{sec:three-col}

Suppose that a graph $(V,E)$ with vertices $V$ and
edges $E$ is given as a $\tau$-structure
with $\tau = \{ e\}$, i.e., $e$ is the binary edge relation. This graph is 3-colorable, iff there exists a partition of $V$ into three sets $\RR$, $\GG$, 
$\BB$, s.t.\
no two adjacent vertices $v_1, v_2 \in V$ 
are in the same set
$\RR$, $\GG$, or $\BB$. This criterion can be easily expressed
by an MSO-sentence, namely

\begin{eqnarray*}
\phi & \equiv & \exists R
\exists G \exists B [
Partition(R,G,B) \wedge \mbox{}
\forall v_1 \forall v_2 
[e(v_1,v_2) 
\ra
\\
& &(\neg R(v_1) \vee \neg R(v_2)) \wedge
(\neg G(v_1) \vee \neg G(v_2)) \wedge
(\neg B(v_1) \vee \neg B(v_2)) 
\mbox{ with } \\
Partition(R,G,B) & \equiv & 
\forall v [[R(v) \vee G(v) \vee B(v)] \wedge
\\
& & (\neg R(v) \vee \neg G(v)) \wedge
(\neg R(v) \vee \neg B(v)) \wedge
(\neg G(v) \vee \neg B(v))].
\end{eqnarray*}

\nop{******************************
\[
\begin{array}{lll}
\phi & \equiv & \exists R
\exists G \exists B [
Partition(R,G,B) \wedge \mbox{}
\forall v_1 \forall v_2 
[e(v_1,v_2) 
\ra
\\
& &(\neg R(v_1) \vee \neg R(v_2)) \wedge
(\neg G(v_1) \vee \neg G(v_2)) \wedge
(\neg B(v_1) \vee \neg B(v_2)) 
]]
\end{array}
\]

with 

\[
\begin{array}{lll}
Partition(R,G,B) & \equiv & 
\forall v [[R(v) \vee G(v) \vee B(v)] \wedge
\\
& & (\neg R(v) \vee \neg G(v)) \wedge
(\neg R(v) \vee \neg B(v)) \wedge
(\neg G(v) \vee \neg B(v))]
\end{array}
\]
******************************}%nop

Suppose that a graph $(V,E)$ 
together with a tree decomposition $\TT$ of width $w$
is given as 
a $\tauTD$-structure with 
$\tauTD = \{e,
\root, \leaf, 
\child_1, \child_2, \bag \}$.
In Figure~\ref{fig:Three-Col}, we describe a datalog program which takes such a 
$\tauTD$-structure as input and decides if the 
graph thus represented is 3-colorable.

\begin{figure}[t]
\begin{center}

\noindent
%\fbox{\begin{minipage}{3.40in}
\fbox{\begin{minipage}{5.8in}

{\bf Program} 3-Colorability

\small

%\smallskip
\vspace{2.2pt}

\noindent
/* leaf node. */

$\solve(s,R,G,B)$ $\la$  $\leaf(s)$, $\bag(s,X)$, 
$\partition(s,R,G,B)$,
$\allowed(s,R)$,
$\allowed(s,G)$, $\allowed(s,B)$.

%\smallskip
\vspace{2.2pt}

\noindent
/* element introduction node. */

$\solve(s,R \uplus \{v\},G,B)$ $\la$  $\bag(s,X\uplus \{v\})$, 
$\child_1(s_1,s)$, 
$\bag(s_1,X)$, 
$\solve(s_1,R,G,B)$
$\allowed(s,R \uplus \{v\})$.

\vspace{0.3pt}

$\solve(s,R,G\uplus \{v\},B)$ $\la$  $\bag(s,X\uplus \{v\})$, 
$\child_1(s_1,s)$, 
$\bag(s_1,X)$, 
$\solve(s_1,R,G,B)$
$\allowed(s,G \uplus \{v\})$.

\vspace{0.3pt}

$\solve(s,R,G,B\uplus \{v\})$ $\la$  $\bag(s,X\uplus \{v\})$, 
$\child_1(s_1,s)$, 
$\bag(s_1,X)$, 
$\solve(s_1,R,G,B)$
$\allowed(s,B \uplus \{v\})$.

%\smallskip
\vspace{2.2pt}

\noindent
/* element removal node. */ 

$\solve(s,R,G,B)$ $\la$  $\bag(s,X)$, 
$\child_1(s_1,s)$, 
$\bag(s_1,X\uplus \{v\})$, 
$\solve(s_1,R\uplus \{v\},G,B)$.

\vspace{0.3pt}

$\solve(s,R,G,B)$ $\la$  $\bag(s,X)$, 
$\child_1(s_1,s)$, 
$\bag(s_1,X\uplus \{v\})$, 
$\solve(s_1,R,G\uplus \{v\},B)$.

\vspace{0.3pt}

$\solve(s,R,G,B)$ $\la$  $\bag(s,X)$, 
$\child_1(s_1,s)$, 
$\bag(s_1,X\uplus \{v\})$, 
$\solve(s_1,R,G,B\uplus \{v\})$.

%\smallskip
\vspace{2.2pt}

\noindent
/* branch node. */

$\solve(s,R,G,B)$ $\la$  $\bag(s,X)$, 
$\child_1(s_1,s)$, $\child_2(s_2,s)$,
$\bag(s_1,X)$, 
$\bag(s_2,X)$, 
$\solve(s_1,R,G,B)$,

~~~~~
$\solve(s_2,R,G,B)$.

%\smallskip
\vspace{2.2pt}

\noindent
/* result (at the root node). */

$\succ \la \root(s)$, $\solve(s,R,G,B)$.

\normalsize

\end{minipage}}

\caption{3-Colorability Test.}
\label{fig:Three-Col}
\end{center}
\end{figure}

Some words on the notation used in this program are in order:
We are using lower case letters $s$ and $v$ (possibly with subscripts) 
as datalog variables for a
single node in $\TT$ and for a single vertex in $V$, 
respectively. In contrast, upper case letters $X$, 
$R$, $G$, and $B$ 
are used as datalog variables denoting sets of vertices.
Note that these sets are not sets in the general sense, 
since their cardinality is
restricted by the size $w+1$ of the bags, where
$w$ is a fixed constant. Hence,  
these ``fixed-size'' sets can be simply  implemented
by means of 
$k$-tuples with $k \leq (w+1)$ over $\{0,1\}$.
For the sake of readability, 
we are using non-datalog expressions 
with the set operator
$\uplus$ (disjoint union).
For the fixed-size sets 
under consideration here, one could, of course, easily
replace this operator by ``proper'' datalog expressions
of the form $\d-union(R,\{v\}, R')$.

It is convenient to 
introduce the following notation. 
Let $G = (V,E)$ be the input graph
with tree decomposition $\TT$.
For any node $s$ in $\TT$, 
we write as usual $\Ts$ to denote the subtree of $\TT$ rooted at $s$. Moreover, we write
$V(s)$ and $V(\Ts)$ to denote the vertices in the bag of $s$ respectively in any bag in $\Ts$. 

Our 3-Colorability-program checks if $G$ is 
3-colorable via the criterion mentioned above, i.e., 
there exists a partition of $V$ into three sets $\RR$, $\GG$, 
$\BB$, s.t.\
no two adjacent vertices $v_1, v_2 \in V$ 
are in the same set
$\RR$, $\GG$, or $\BB$.

At the heart of this program is the intensional predicate 
$\solve(s, R,G,B)$ with the following intended meaning:
$s$ denotes a node in $\TT$ and 
$R$, $G$, $B$ are the projections of 
$\RR$, $\GG$, $\BB$
onto $V(s)$. 
For 
all values
$s, R,G,B$,
the ground fact 
$\solve(s, R,G,B,)$
shall be in the 
least fixpoint 
of the program plus the input structure, iff the 
following condition holds:

\smallskip

\noindent
{\sc Property A.} 
There exist extensions $\ER$ of $R$, 
$\EG$ of $G$, and $\EB$ of $B$ to 
$V(\Ts)$, s.t.\ 
\begin{enumerate}
\setlength{\itemsep}{-0.5mm}
\item $\ER$, $\EG$, and $\EB$ form a partition of $V(\Ts)$ and
\item no two adjacent vertices $v_1, v_2 \in V(\Ts)$
are in the same set
$\ER$, $\EG$, or $\EB$.
\end{enumerate}
%
%\smallskip
%
\noindent
In other words, $\ER$, $\EG$, and $\EB$ is a valid
3-coloring of the vertices in $V(\Ts)$ and 
$R$, $G$, and $B$ are the projections of 
$\ER$, $\EG$, and $\EB$ onto $V(s)$.

The main task of the program is the computation of all facts
$\solve(s, R, G, B)$ via a bottom-up traversal of the
tree decomposition. The other predicates have the following meaning: 
\begin{itemize}
\setlength{\itemsep}{-0.5mm}
\item $\partition(s,R,G,B)$ is in the 
least fixpoint 
iff $R$, $G$, $B$ is a partition of the bag $X$ at
node $s$ in the tree decomposition. 
\item 
$\allowed(s,X)$ is in the 
least fixpoint iff $X$ contains no 
adjacent vertices $v_1, v_2$.
\end{itemize}
Recall that the cardinality of the sets $X$, $R$, $G$, $B$ occurring as arguments of $\partition$ and
$\allowed$ is bounded by
the fixed constant $w + 1$. 
In fact, both the  $\partition$ predicate and the 
$\allowed$ predicate can  be 
treated as extensional predicates by computing all
facts $\partition(s,R,G,B)$ and $\allowed(s,X)$ 
for each node $s$ in $\TT$ as part of the computation of the tree decomposition. This additional computation also fits into the linear time bound.

The intuition of the rules with the \solve-predicate in the head 
is now clear: At the {\em leaf nodes\/}, the 
program generates ground facts $\solve(s,R,G,B)$ for all possible partitions
of the bag $X$ at $s$, such that none of the sets $R, G, B$ contains 
two adjacent vertices. The three rules for 
{\em element introduction nodes\/} distinguish the three cases if 
the new vertex $v$ is added to $R$, $G$, or $B$, respectively. 
Of course, by the \allowed-atom in the body of these 3 rules, the attempt to 
add $v$ to any of the sets $R$, $G$, or $B$ may fail.
The three rules for 
{\em element removal nodes\/} distinguish the three cases if 
the removed vertex was in $R$, $G$, or $B$, respectively. The rule for 
{\em branch nodes\/} combines \solve-facts with identical values of 
$(R, G, B)$ at the child nodes $s_1$ and $s_2$ to the corresponding
\solve-fact at $s$.

In summary, the 3-colorability-program has the following properties.

\begin{theorem}
\label{theo:three-col}
The datalog program 
in Figure~\ref{fig:Three-Col} decides the 
3-Colorability problem, i.e.,
the fact ``success'' is in the 
least fixpoint of this program  plus the input $\tauTD$-structure $\AATD$ 
iff $\AATD$ 
encodes a 
3-colorable graph
$(V,E)$.
Moreover, for any graph $(V,E)$ with treewidth $w$, 
the computation of the $\tauTD$-structure $\AATD$ and the evaluation of 
the program can be done in time
${\cal O}(f(w) * |(V,E)|)$ for some function $f$.
\end{theorem}

\begin{proof}
By the above considerations, it is clear that the predicate $\solve$ 
indeed has the meaning
described by Property A. A formal proof of this fact 
by structural
induction on $\TT$ is immediate and therefore omitted here.
Then the rule with head $\succ$ reads as follows:
$\succ$ is in the 
least fixpoint, iff $s$ denotes the root of $\TT$
and there 
exist extensions $\ER$, $\EG$, and $\EB$ 
of $R,G,B$ to $V(\Ts)$ (which is identical to $V$ in case of the root node $s$), s.t.\
$\ER$, $\EG$, and $\EB$ 
is a valid
3-coloring of the vertices in $V(\Ts) = V$.

For the linear time data complexity, the crucial observation
is that our program in Figure~\ref{fig:Three-Col} is essentially
a succinct representation of a quasi-guarded monadic datalog program.
For instance, in the atom 
$\solve(s, R, G, B)$, 
the  sets 
$R, G, B$ 
are subsets of the bag of $s$. Hence, 
each combination 
$R, G, B$ 
could be represented by 3 subsets 
$r_1, r_2, r_3$ over $\{0, \dots, w\}$ referring to indices of 
elements in the bag of $s$. 
Recall that $w$ is a fixed constant. Hence, 
$\solve(s, R, G, B)$ 
is simply a succinct representation of 
constantly
% (= exponentially w.r.t.\ $w$) 
many monadic predicates of the form
$\solve_{ \langle r_1, r_2, r_3\rangle } (s)$. The quasi-guard in each rule 
can thus be any atom with argument $s$, e.g., $\bag(s,X)$
(possibly extended by a disjoint union with $\{v\}$).
Thus, the linear time bound follows immediately from 
Theorem~\ref{theo:mondatalog-complexity}.
$\punto$
\end{proof}

\noindent
{\em Discussion.}
Let us briefly compare the monadic program 
constructed
in the proof of Theorem~\ref{theo:MSO-to-Datalog} with 
the 3-Colorability program in Figure~\ref{fig:Three-Col}.
Actually, since we are dealing with a decision problem here, we
only look at the bottom-up construction in the proof of Theorem \ref{theo:MSO-to-Datalog}, since the top-down construction 
is not needed for a 0-ary target formula $\phi()$.
As was already mentioned in the proof of 
Theorem~\ref{theo:three-col}, 
the atoms $\solve(s, R, G, B)$ can be 
thought of as a succinct representation for
atoms  of the form 
$\solve_{ \langle r_1, r_2, r_3\rangle } (s)$. 
Now the question naturally arises where
the type $\theta$ of some node $s$ from the proof of Theorem \ref{theo:MSO-to-Datalog}
is present in the 3-Colorability program. 
A first tentative answer is that this type
essentially corresponds to the set 
$R(s) = \{\langle r_1, r_2, r_3\rangle \mid 
\solve_{ \langle r_1, r_2, r_3\rangle } (s)$ is in the least fixpoint$\}$.
However, there are two significant aspects which distinguish
our 3-Colorability program from merely a succinct representation of the 
type transitions encoded in the 
monadic datalog program of Theorem \ref{theo:MSO-to-Datalog}: 
\begin{enumerate}
\item 
By Property A, we are only interested in the types of those structures
which -- in principle --  could be extended in bottom-up direction 
to a structure representing a satisfiable propositional formula. Hence, 
in contrast to the construction in the proof of Theorem~\ref{theo:MSO-to-Datalog}, 
our 3-Colorability program does clearly not 
keep track of all possible types that the substructure induced by some
tree decomposition $\Ts$
may possibly have.
\item  
$R(s) = \{\langle r_1, r_2, r_3\rangle \mid 
\solve_{ \langle r_1, r_2, r_3\rangle } (s)$ is in the least fixpoint$\}$
does not exactly correspond to the type of $s$. Instead, it only describes 
the crucial properties of the type. Thus, the 3-Colorability program
somehow ``aggregates'' several types from the proof of Theorem \ref{theo:MSO-to-Datalog}.
\end{enumerate}
These two properties ensure that the 3-Colorability program is much shorter 
than the program in the proof of Theorem~\ref{theo:MSO-to-Datalog} and that
the difference between these two programs is not just due to the 
succinct representation of a monadic program by a non-monadic one.
The deeper reason of this improvement is that we take the target MSO formula $\phi$
(namely, the characterization of 3-Colorability) into account for the 
entire construction of the datalog program
in Figure~\ref{fig:Three-Col}. In contrast, the rules describing
the type-transitions in the proof of Theorem~\ref{theo:MSO-to-Datalog}
for a bottom-up traversal of the tree decomposition
are fully generic. Only the rules with head predicate $\phi$ are
specific to the actual target MSO formula $\phi$.

\subsection{The Primality Decision Problem}
\label{sec:primality}

Recall from Section~\ref{sec:treewidth} that we 
represent a relational schema $(R,F)$
as a $\tau$-structure with $\tau = \{\fd, \att, \lh,$ $\rh
 \}$. 
Moreover, recall that, in Section~\ref{sec:Work},  
we consider normalized tree decompositions with 
{\em element removal nodes} and 
{\em element introduction nodes} 
rather than {\em element replacement nodes} as in 
Section~\ref{sec:treewidth}. 
With our representation of relational schemas $(R,F)$ as finite structures, 
the domain elements are the attributes and FDs in $(R,F)$. 
Hence, in total, the former element replacement nodes give rise 
to four kinds of nodes, namely, attribute removal nodes, 
FD removal nodes, attribute introduction nodes, and
FD introduction nodes. Moreover, 
we now  consider the bags as a {\em pair of sets\/} $(\At, \Fd)$, where 
$\At$ is a set attributes and
$\Fd$ 
is a set of FDs.
Again, 
we may delete permutation nodes from the tree decomposition.
Finally, it will greatly simplify the presentation of 
our datalog program if we require that, whenever an FD $f \in F$ 
is contained in a bag of the tree decomposition, then 
the attribute $\rhs(f)$ is as well. In the worst-case, this may 
double the width of the resulting decomposition.

Suppose that a schema $(R,F)$ 
together with a tree decomposition $\TT$ of width $w$
is given as 
a $\tauTD$-structure with 
$\tauTD = \{\fd, \att, \lh, \rh, 
%\} \cup \{ 
\root, \leaf, 
\child_1, \child_2, \bag \}$.
In Figure~\ref{fig:Primality}, we describe a datalog program,
where the input is given as an attribute $a \in R$ and 
a $\tauTD$-structure, s.t.\ $a$ occurs in the bag at the root of
the tree decomposition.

\begin{figure}[h!]
\begin{center}

\noindent
%\fbox{\begin{minipage}{3.40in}
\fbox{\begin{minipage}{5.80in}

{\bf Program} PRIMALITY 

\small

%\smallskip
\vspace{2.2pt}

\noindent
/* leaf node. */

$\solve(s,Y,\FY, \C,\DC,\FC)$ $\la$  $\leaf(s)$, $\bag(s,\At,\Fd)$, 
$Y \cup \C = \At$, $Y \cap \C = \emptyset$, 
$\outside(\FY, Y,\At,\Fd )$,  

~~~
$\FC \subseteq \Fd$, $\consistent(\FC,\C)$, 
$\DC = \{\rhs(f) \mid f \in \FC\}$,
$\DC \subseteq \C$.

%\smallskip
\vspace{2.2pt}

\noindent
/* attribute introduction node. */

$\solve(s, Y \uplus \{b\},\FY,\C,\DC,\FC)$ $\la$  $\bag(s,\At \uplus \{b\},\Fd)$, 
$\child_1(s_1,s)$, 
$\bag(s_1,\At, \Fd)$, 

~~~
$\solve(s_1, Y,\FY,\C,\DC,\FC)$.

\vspace{0.3pt}

$\solve(s, Y,\FY,\C \uplus \{b\},\DC,\FC)$ $\la$  $\bag(s,\At \uplus \{b\},\Fd)$, 
$\child_1(s_1,s)$, 
$\bag(s_1,\At, \Fd)$, 

~~~
$\consistent(\FC,\C \uplus \{b\})$,
$\solve(s_1, Y,\FYone,\C,\DC,\FC)$,
$\outside(\FYtwo, Y,\At, \Fd)$,
$\FY = \FYone \cup \FYtwo$.

%\smallskip
\vspace{2.2pt}

\noindent
/* FD introduction node. */

$\solve(s, Y,\FY,\C,\DC,\FC)$ $\la$  $\bag(s,\At,\Fd \uplus \{f\})$, 
$\child_1(s_1,s)$, 
$\bag(s_1,\At, \Fd)$, 
$\rh(b,f)$, $b \in Y$,

~~~
$\solve(s_1, Y,\FY,\C,\DC,\FC)$.

\vspace{0.3pt}

$\solve(s, Y, \FY, \C, \DC \uplus \{b\},\FC \uplus \{f\} )$ 
   $\la$  $\bag(s,\At,\Fd \uplus \{f\})$, 
$\child_1(s_1,s)$, 
$\bag(s_1,\At, \Fd)$, 
$\rh(b,f)$,

~~~
$b \in \C$,
$\solve(s_1, Y,\FYone,\C,\DC,\FC)$,
$\consistent(\{f\},\C)$,
$\outside(\FYtwo, Y,\At, \{f\})$,
$\FY = \FYone \cup \FYtwo$.

\vspace{0.3pt}

$\solve(s, Y, \FY, \C, \DC,\FC )$ 
   $\la$  $\bag(s,\At,\Fd \uplus \{f\})$, 
$\child_1(s_1,s)$, 
$\bag(s_1,\At, \Fd)$, 
$\rh(b,f)$, $b \in \C$,

~~~
$\solve(s_1, Y,\FYone,\C,\DC,\FC)$,
$\outside(\FYtwo, Y,\At, \{f\})$,
$\FY = \FYone \cup \FYtwo$.

%\smallskip
\vspace{2.2pt}

\noindent
/* attribute removal node. */ 

$\solve(s, Y, \FY, \C, \DC,\FC )$ 
   $\la$  $\bag(s,\At,\Fd)$, 
$\child_1(s_1,s)$, 
$\bag(s_1,\At \uplus \{b\}, \Fd)$, 

~~~
$\solve(s_1, Y \uplus \{b\}, \FY, \C,\DC,\FC)$.

\vspace{0.3pt}

$\solve(s, Y, \FY, \C, \DC,\FC )$ 
   $\la$  $\bag(s,\At,\Fd)$, 
$\child_1(s_1,s)$, 
$\bag(s_1,\At \uplus \{b\}, \Fd)$,

~~~
$\solve(s_1, Y, \FY, \C \uplus \{b\},\DC \uplus \{b\},\FC)$.

%\smallskip
\vspace{2.2pt}

\noindent
/* FD removal node. */

$\solve(s, Y, \FY, \C, \DC,\FC )$ $\la$  $\bag(s,\At,\Fd)$, 
$\child_1(s_1,s)$, 
$\bag(s_1,\At, \Fd \uplus \{f\})$,  $\rh(b,f)$, $b \in Y$,

~~~~
$\solve(s_1, Y, \FY, \C, \DC,\FC )$.

\vspace{0.3pt}

$\solve(s, Y, \FY, \C, \DC,\FC )$ $\la$  $\bag(s,\At,\Fd)$, 
$\child_1(s_1,s)$, 
$\bag(s_1,\At, \Fd \uplus \{f\})$,  $\rh(b,f)$, $b \in \C$,

~~~~
$\solve(s_1, Y, \FY \uplus \{f\}, \C, \DC,\FC \uplus \{f\})$.

\vspace{0.3pt}

$\solve(s, Y, \FY, \C, \DC,\FC )$ $\la$  $\bag(s,\At,\Fd)$, 
$\child_1(s_1,s)$, 
$\bag(s_1,\At, \Fd \uplus \{f\})$,  $\rh(b,f)$, $b \in \C$,

~~~~
$\solve(s_1, Y, \FY \uplus \{f\}, \C, \DC,\FC )$, $f \not\in \FC$.

%\smallskip
\vspace{2.2pt}

\noindent
/* branch node. */

$\solve(s, Y, \FYone \cup \FYtwo, \C,\DC_1 \cup \DC_2,\FC)$ 
$\la$  $\bag(s,\At ,\Fd)$, 
$\child_1(s_1,s)$, $\bag(s_1,\At ,\Fd)$, 

~~~
$\child_2(s_2,s)$, $\bag(s_2,\At ,\Fd)$, 
$\solve(s_1, Y, \FYone, \C,\DC_1,\FC)$, 

~~~
$\solve(s_2, Y, \FYtwo, \C,\DC_2,\FC)$, 
$\unique(\DC_1,\DC_2,\FC)$.

%\smallskip
\vspace{2.2pt}

\noindent
/* result (at the root node). */

$\succ \la \root(s), \bag(s,\At,\Fd)$, 
$a \in \At$, 
$\solve(s, Y, \FY, \C, \DC,\FC )$, 
$a \not\in Y$, 

~~~
$\FY = \{ f \in \Fd \mid \rhs(f) \not\in Y\}$, 
$\DC =  \C \setminus \{a\}$.

\normalsize

\end{minipage}}

\caption{Primality Test.}
\label{fig:Primality}
\end{center}
\end{figure}

Analogously to 
Section~\ref{sec:three-col}, 
we are using lower case letters $s$, $f$, and $b$ (possibly with subscripts) 
as datalog variables for a
single node in $\TT$, for a single FD, 
or for a single attribute in $R$, 
respectively. Upper case letters are used as datalog variables denoting sets of attributes
(in the case of $Y, \At, \C, \DC$) or sets of FDs
(in the case of $\Fd, \FY, \FC$). In addition, $\C$ is considered as an
ordered set (indicated by the superscript $o$). When we write
$\C \uplus \{b\}$, we mean that $b$ is arbitrarily ``inserted'' 
into $\C$, leaving the order of the remaining elements unchanged.
Again, the cardinality of these (ordered) sets is
restricted by the size $w+1$ of the bags, where
$w$ is a fixed constant.
In addition to $\uplus$ (disjoint union)
we are now also using the set operators
$\cup$,
$\cap$, $\subseteq$, and $\in$. For the fixed-size (ordered) sets 
under consideration here, one could, of course, easily
replace these operators by ``proper'' datalog expressions. 
Moreover, for the 
input schema $(R,F)$ 
with tree decomposition $\TT$
we use the following notation:
We write 
$\FD(s)$ to denote the FDs in the bag of $s$ and
$\FD(\Ts)$ to denote the FDs that occur in any bag in $\Ts$. 
Analogously, we write 
$\Att(s)$ and $\Att(\Ts)$  as a short-hand for the attributes occurring 
in the bag of $s$ respectively in any bag in $\Ts$.

Our PRIMALITY-program checks the primality of $a$ by 
via the criterion used for the MSO-characteri\-zation in
Example~\ref{bsp:primality-MSO}: We have to 
search for an attribute set $\YY \subseteq R$, s.t.\ $\YY$ is closed
w.r.t.\ $F$ (i.e., $\YY^+ = \YY$), $a \not\in \YY$ and
$(\YY \cup \{a\})^+ = R$, i.e., 
$\YY \cup \{a\}$ is a superkey but $\YY$ is not.

At the heart of our PRIMALITY-program is the intensional predicate 
$\solve(s, Y, \FY, \C, \DC,\FC )$ with the following intended meaning:
$s$ denotes a node in $\TT$. 
$Y$ (resp.\ $\C$) 
is the projection of 
$\YY$ (resp. of $R \setminus \YY$) onto $\Att(s)$. 
We consider $R \setminus \YY$ as ordered w.r.t.\ an appropriate
derivation sequence of $R$ from $\YY \cup \{a\}$, i.e., suppose
that $\YY \cup \{A_0\} \ra   \YY \cup \{A_0, A_{1} \} \ra 
\YY \cup \{A_{0}, A_{1}, A_{2} \} \ra
\dots \ra \YY \cup \{A_{0}, A_{1}, \dots, A_{n}  \}$, s.t.\
$A_0 = a$ and  $\YY \cup \{A_{0}, A_{1}, \dots, A_{n}\} = R$.
W.l.o.g., the $A_i$'s may be assumed to be pairwise distinct.
Then for any two $i \neq j$,
we simply set $A_i < A_j$ iff $i < j$.
By the connectedness condition on $\TT$, our datalog program 
ensures that the order on each subset $\C$ of 
$R \setminus \YY$ is consistent with the overall ordering.

The argument $\FY$ of the \solve-predicate is used to 
guarantee that $\YY$ is indeed closed. 
Informally, $\FY$ contains those
FDs in $\FD(s)$ for which we
have already verified (on the bottom-up traversal of the tree decomposition) 
that they do not constitute a contradiction
with the closedness of $\YY$. In other words, either 
$\rhs(f) \in \YY$ or there exists an attribute in 
$\lhs(f) \cap \At(\Ts)$ which is not in 
$\YY$.

The arguments $\DC$ and $\FC$ of the \solve-predicate are used to ensure that 
$(\YY \cup \{a\})^+ = R$ indeed holds: 
The intended meaning of the set $\FC$ 
is that it contains those FDs in $\FD(s)$ which are used in the 
above derivation sequence. 
Moreover, $\DC$ contains
those attributes from $\Att(s)$ for which we have already
shown that they can be derived from $\YY$ plus smaller 
atoms in $\C$. 

More precisely, for 
all values
$s, Y, \FY, \C, \DC,\FC$,
the ground fact 
$\solve(s, Y, \FY, \C, \DC,\FC )$
shall be in the 
least fixpoint 
of the program plus the input structure, iff the 
following condition holds:

\smallskip

\noindent
{\sc Property B.} 
There exist extensions $\EY$ of $Y$ and $\EC$ of $\C$ to 
$\Att(\Ts)$ and an extension $\EFC$ of $\FC$ to 
$\FD(\Ts)$, s.t.\ 
\begin{enumerate}
\setlength{\itemsep}{-0.5mm}
\item $\EY$ and $\EC$ form a partition of $\Att(\Ts)$,
\item $\forall f \in \FD(\Ts) \setminus \FD(s)$, if
$\rhs(f) \not\in \EY$, then $\lhs(f) \not\subseteq \EY$.
Moreover, 
$\FY = 
\{f \in \FD(s) \mid \rhs(f) \not\in \EY$
and $\lhs(f) \cap \Att(\Ts) \not\subseteq \EY\}$.
\item $\forall f \in \EFC$, 
$f$ is consistent with 
the order on $\EC$, i.e., $\forall f \in \EFC$:
$\rhs(f) \in \EC$ and 
$\forall b \in \lhs(f) \cap \EC$: $b < \rhs(f)$ holds.
\item $  \DC \cup \EC \setminus \Att(s) = 
\{\rhs(f) \mid f \in \EFC\}$,
\end{enumerate}
%
%\smallskip
%
\noindent
The main task of the program is the computation of all facts
$\solve(s, Y, \FY, \C, \DC,\FC )$ by means of a bottom-up traversal of the
tree decomposition. The other predicates have the following meaning: 
\begin{itemize}
\setlength{\itemsep}{-0.5mm}
\item $\outside(\FY, Y,\At,\Fd )$ is in the 
least fixpoint 
iff $\FY = 
\{f \in \Fd \mid \rhs(f) \not\in Y$ 
and  $\lhs(f) \cap \At \not\subseteq Y\}$, i.e., 
for every $f \in \FY$, $\rhs(f)$ is outside $Y$ but
this will never conflict with the closedness of
$Y$ because $\lhs(f)$ contains an attribute
from outside $Y$.
\item 
$\consistent(\FC,\C)$ is in the 
least fixpoint 
iff $\forall f \in \FC$ we have $\rhs(f) \in \C$ and
$\forall b \in \lhs(f) \cap \C$:
$b < \rhs(f)$, i.e., the FDs in 
$\FC$  are only used to derive greater attributes from
smaller ones (plus attributes from $\YY$).
\item  The fact $\unique(\DC_1,\DC_2,\FC)$ 
is in the 
least fixpoint iff the condition 
$\DC_1 \cap \DC_2 = \{ b \mid b = \rhs(f)$ for some $f\in \FC\}$ holds. 
The $\unique$-predicate is only used in the body of the rule for branch nodes.
Its purpose is to avoid that an attribute in $R \setminus \YY$ is derived via two different FDs in the two subtrees at the child nodes of the branch node.

\item 
The 0-ary predicate 
$\succ$ indicates if the fixed attribute $a$ is
prime in the schema encoded by the input structure.
\end{itemize}
The PRIMALITY-program has the following properties.

\begin{lemma}
\label{lemma:solve-correct}
The \solve-predicate has the intended meaning described above, i.e., 
for 
all values
$s$, $Y$, $\FY$, $\C$, $\DC$, $\FC$,
the ground fact 
$\solve(s, Y, \FY, \C, \DC,\FC )$
is in the 
least fixpoint 
of the PRIMALITY-program plus the input structure, iff Property B holds.
\end{lemma}

\begin{proofsketch}
The lemma can be shown by structural
induction on $\TT$. We restrict ourselves here
to outlining the ideas underlying the various rules 
of the PRIMALITY-program. The induction itself is then obvious and therefore omitted.

\smallskip

\noindent
(1) {\em leaf nodes.\/} The rule for a leaf node $s$ realizes
two ``guesses'' so to speak: (i) a partition of $\At(s)$ 
into $Y$ and $\C$ together with an ordering on $\C$ and
(ii) the subset $\FC \subseteq \Fd(s)$ of FDs which are 
used in the derivation sequence of $R \setminus \YY$
from $\YY \cup \{a\}$. The remaining variables are thus fully determined: $\FY$ is determined via the \outside-predicate, while $\DC$ is determined via the
equality $\DC = \{\rhs(f) \mid f \in \FC\}$. 
Finally the body of the rule contains the checks
$\consistent(\FC,\C )$ and
$\DC \subseteq \C$ to make sure that (at least at the 
leaf node $s$) the ``guesses'' are allowed.

\smallskip

\noindent
(2) {\em attribute introduction node.\/}
The two rules are used to distinguish 2 cases whether the new attribute $b$ is added to $Y$ or to $\C$. If $b$ is added to $Y$ then all arguments of the \solve-fact at the
child node $s_1$ of $s$ remain unchanged at $s$. In contrast, if $b$ is inserted into $\C$ then the following actions are required: 

The atom $\consistent(\FC,\C \uplus \{b\})$ makes sure that the rules in $\FC$  
are consistent with the ordering of $\C$, i.e., it must not happen that 
the new attribute $b$ occurs in $\lhs(f)$ for some $f \in FC$, s.t.\
$b > \rhs(f)$ holds. 

The new attribute $b$ outside $Y$ may possibly allow
us to verify for some additional FDs that they do not contradict the closedness of $\YY$. The atom $\outside(\FYtwo, Y,\At, \Fd)$ determines the set 
$\FYtwo$ which contains all FDs with $\rhs(f) \not\in Y$ but with some attribute from $\C$ (in particular, the new attribute $b$) 
in $\lhs(f)$.

Recall that we are requiring  that, whenever an FD $f \in F$ 
is contained in a bag of the tree decomposition, then 
the attribute $\rhs(f)$ is as well. Hence, since the attribute $b$ has
just been introduced on our bottom-up traversal of the tree decomposition, 
we can be sure that $b$ does not occur on the right-hand side of 
any FD in the bag of $s$. Thus, $\DC$ is not affected by the transition from
$s_1$ to $s$.

\smallskip

\noindent
(3) {\em FD introduction node.\/}
The three rules distinguish, in total, 3 cases: First, does 
$\rhs(f) \in Y$ or $\rhs(f) \in \C$ hold? (Recall that we assume that every bag containing some FD also contains the right-hand side of this FD.) The latter case is then further divided into the subcases if $f$ is used for the derivation 
of $R \setminus \YY$ or not. The first rule deals with the case $\rhs(f) \in Y$. Then all arguments of the \solve-fact at the
child node $s_1$ of $s$ remain unchanged at $s$. 

The second rule addresses the case that $\rhs(f) \in \C$ and $f$ is used for the derivation of $R \setminus \YY$.  
Then the attribute $\rhs(f)$ is added to $\DC$. The disjoint union makes sure that this attribute has not yet been derived by another rule with the same right-hand side. 
The atom $\consistent(\FC,\C \uplus \{b\})$
is used to check the consistency of $f$ with the ordering of $\C$.
The atom $\outside(\FYtwo, Y,\At, \Fd)$ is used to check if $f$ may be added to $\FY$, i.e., if some attribute in $\lhs(f)$ is in $\C$.

The third rule refers to the case that 
$\rhs(f) \in \C$ and $f$ is {\em not\/} 
used for the derivation of $R \setminus \YY$.
Again, the atom $\outside(\FYtwo, Y,\At, \Fd)$ is used to check if $f$ may be added to $\FY$.

\smallskip

\noindent
(4) {\em attribute removal node.\/}
The two rules are used to distinguish 2 cases whether the attribute $b$ was in $Y$ or in $\C$. If $b$ was in $Y$ then all arguments of the \solve-fact at the
child node $s_1$ of $s$ remain unchanged at $s$. In contrast, if $b$ 
was in $\C$ then we have to check (by pattern matching with the 
fact $\solve (s_1, \dots, \DC \uplus \{b\}, \dots)$) that 
a rule $f$ for deriving $b$ has already been found. Recall that, on our bottom-up traversal of $\TT$, when we first encounter an attribute $b$, 
it is either added to $Y$ or $\C$. If $b$ is added to $\C$ then we eventually have to determine the FD by which $b$ is derived. Hence, initially, $b$ is in $\C$ but not in $\DC$. However, when $b$ is finally removed from the bag
then its derivation must have been verified. 
The arguments $Y$, $\FY$, and $\FC$ are of course not affected by this
attribute removal.

\smallskip

\noindent
(5) {\em FD removal node.\/}
Similarly to the FD introduction node, we  distinguish, in total, 3 cases. 
If $\rhs(f) \in Y$ then all arguments of the \solve-fact at the
child node $s_1$ of $s$ remain unchanged at $s$. If 
$\rhs(f) \in \C$ then we further distinguish the subcases if $f$ is used for the derivation 
of $R \setminus \YY$ or not. The second and third rule refer two these two subcases. The action carried out by these two rules is the same, namely it has to be checked
(by pattern matching with the 
fact $\solve (s_1, \dots, \FY \uplus \{f\}, \dots)$) that 
$f$ does not constitute a contradiction with the 
closedness of $\YY$. In other words, since $\rhs(f) \in \C$, we must have encountered (on our bottom-up traversal of $\TT$) an attribute
in $\lhs(f) \not\in \YY$.

\smallskip

\noindent
(6) {\em branch node.\/}
Recall that a branch node $s$ and its two child nodes $s_1$ and $s_2$ have identical bags by our notion of normalized tree decompositions. 
The argument of the \solve-fact at $s$ is then determined from the arguments at 
$s_1 $ and $s_2$ as follows: The arguments $Y$ and $\C$ must have the same value at all three nodes $s$, $s_1$, and $s_2$. Likewise, $\FC$ (containing 
the FDs from the bags at these nodes which are used in the derivation of 
$R \setminus \YY$) must be identical. In contrast, $\FY$ and $\DC$ are obtained
as the union of the corresponding arguments in the \solve-facts at the 
child nodes $s_1$ and $s_2$, i.e., it suffices to verify at one of the child nodes $s_1$ or $s_2$ that some FD 
does not contradict the closedness of $Y$ and that some attribute in $\C$
is derived by some FD. 

Recall that we define an order on the attributes  in 
$R \setminus \YY$ by means of  some derivation sequence
of $R \setminus \YY$
from $\YY \cup \{a\}$. 
Hence, we   
we have to make sure that every attribute in $R \setminus \YY$ 
is derived only once in this derivation sequence.
In other words, for every $b \in  
R \setminus (\YY\cup \{a\})$, we use exactly one FD $f$ with $\rhs(f) = b$
in our derivation sequence.
The atom
$\unique(\DC_1,\DC_2,\FC)$ 
in the rule body ensures that 
no attribute in $R \setminus \YY$ is derived via two different FDs in the two subtrees at the child nodes of the branch node. 
$\punto$
\end{proofsketch}

\begin{theorem}
\label{theo:primtest}
The datalog program 
in Figure~\ref{fig:Primality} decides the PRIMALITY problem 
for a fixed attribute $a$, i.e.,
the fact ``success'' is in the 
least fixpoint of this program  plus the input $\tauTD$-structure $\AATD$ 
iff $\AATD$ 
encodes a relational
schema $(R,F)$, s.t.\ $a$ is part of a key.
Moreover, for any 
%relational 
schema $(R,F)$ with treewidth $w$, 
the computation of the $\tauTD$-structure $\AATD$ and the evaluation of 
the program can be done in time
${\cal O}(f(w) * |(R,F)|)$ for some function $f$.
\end{theorem}

\begin{proof}
By Lemma \ref{lemma:solve-correct}, 
the predicate $\solve$ indeed has the meaning
according to Property B. Thus, the rule with head $\succ$ reads as follows:
$\succ$ is in the 
least fixpoint, iff $s$ denotes the root of $\TT$,
$a$ is an attribute in the bag at $s$, and 
$Y$ is the projection of the desired attribute set 
$\YY$ onto $\Att(s)$, i.e., (1) 
$\YY$ is closed (this is ensured by the condition
that $\{ f \in \Fd \mid \rhs(f) \not \in Y\} = \FY$),
(2) $a \not \in \YY$ and, finally,
(3) all attributes in $R \setminus (\YY \cup \{a\})$ are
indeed determined by $\YY \cup \{a\}$ (this is
ensured by the condition 
$\DC =  \C \setminus \{a\}$).

The linear time data complexity is due to 
the same argument as in the proof of 
Theorem \ref{theo:three-col}: 
our program in Figure~\ref{fig:Primality} is essentially
a succinct representation of a quasi-guarded monadic datalog program.
For instance, in the atom 
$\solve(s, Y, \FY, \C, \DC,\FC )$, 
the (ordered) sets 
$Y$, $\FY$, $\C$, $\DC$, and $\FC$ 
are subsets of the bag of $s$. Hence, 
each combination 
$Y$, $\FY$, $\C$, $\DC$, $\FC$
could be represented by 5 subsets resp. tuples
$r_1, \dots, r_5$ over $\{0, \dots, w\}$ referring to indices of 
elements in the bag of $s$. 
Recall that $w$ is a fixed constant. Hence, 
$\solve(s, Y, \FY, \C,$ $\DC,\FC )$, 
is simply a succinct representation of 
constantly
% (= exponentially w.r.t.\ $w$) 
many monadic predicates of the form
$\solve_{ \langle r_1, \dots, r_5\rangle } (s)$. The quasi-guard in each rule 
can thus be any atom with argument $s$, e.g., 
$\bag(s,\At,\Fd)$
(possibly extended by a disjoint union with $\{b\}$ or $\{f\}$, 
respectively).
Thus, the linear time bound follows immediately from 
Theorem~\ref{theo:mondatalog-complexity}.
$\punto$
\end{proof}

%%%%%%%%%%%%%%%%%%%%%%%%%%%%%%%%%%%%%%%%%%%%%%%%%%%%%%%%%%%
%%%%%%%%%%%%%% Abduction %%%%%%%%%%%%%%%%%%%%%%%%%%%%%%%%%%
\subsection{The Primality Enumeration Problem}
\label{sec:monadic-prime}

In order to extend the Primality algorithm from the previous section to 
a monadic predicate selecting all prime attributes in a schema, a naive
first attempt might look as follows: one can consider
the tree decomposition $\TT$ as rooted at various nodes, s.t.\ each $a \in R$ is
contained in the bag of one such root node. Then, for each $a$ and corresponding tree decomposition $\TT$, we run the algorithm from 
Figure~\ref{fig:Primality}.
Obviously, this method
has {\em quadratic} time complexity w.r.t.\ the data size. However, 
in this section, we describe a {\em linear\/} time algorithm.

The idea of this algorithm is to implement 
a top-down traversal of the tree decomposition in addition
to the bottom-up traversal realized by the program in 
Figure~\ref{fig:Primality}. For this purpose, we modify 
our notion of {\em normalized\/} tree decompositions in the following 
way: First, any tree decomposition can of course be
transformed in such a way that every attribute $a \in R$ 
occurs in at least one leaf node of $\TT$. 
Moreover, for every branch node $s$ in the tree decomposition, 
we insert a
new node $u$ as new parent of $s$, s.t.\ $u$ and $s$ have identical bags. Hence, together with the two child nodes of $s$, each branch node is
``surrounded'' by three neighboring nodes with identical bags. 
It is thus guaranteed that a branch node always has two child nodes with identical bags, no matter where $\TT$ is rooted. Moreover, this insertion of a new node also implies that the root node of $\TT$ is not a branch node.

We propose the following algorithm for computing a 
mo\-na\-dic predicate $\primality()$, which selects precisely the prime attri\-butes in $(R,F)$. In addition to the predicate $\solve$, whose meaning
was described by Property B in Section~\ref{sec:primality},
we also compute a predicate
\solveD, whose meaning is described by replacing every occurrence
of $\Ts$ in Property B by $\bar{\T}_s$. As the notation \solveDx suggests, the computation of 
\solveDx can be done via a top-down traversal of $\TT$. 
Note that \solveD($s, \dots)$  for a leaf node $s$ of $\T$ is exactly the same as
if we computed $\solve(s,\dots)$ for the tree rooted at $s$. 
Hence, we can define the predicate $\primality()$ as follows.

\smallskip

\begin{center}

\noindent
%\fbox{\begin{minipage}{3.30in}
\fbox{\begin{minipage}{5.40in}

{\bf Program} Monadic-Primality

\smallskip

\small

$\primality(a)$  $\la$ $\leaf(s)$, $\bag(s,\At,\Fd)$, 
$a \in \At$, 
\solveD$(s, Y, \FY, \C, \DC,\FC )$, 
$a \not\in Y$,

~~~~~~
$\FY = \{ f \in \Fd \mid \rhs(f) \not\in Y\}$,
$\DC =  \C \setminus \{a\}$.

\normalsize

\end{minipage}}
\end{center}

\smallskip

By the intended meaning of 
\solveDx and by the properties of
the Primality algorithm in Section \ref{sec:primality}, 
we immediately get the following result.

\begin{theorem}
\label{theo:relevance}
The monadic predicate $\primality()$ as defined above
selects precisely the prime attributes. Moreover,
it can be computed in linear time w.r.t.\ the 
size of the input structure.
\end{theorem}

\section{Implementation and Results}
\label{sec:Results} 

To test our new datalog programs in terms of their scalability with 
a large number of attributes and rules, 
we have implemented the
Primality program from Section~\ref{sec:primality}
in C++. The experiments were conducted on Linux kernel 2.6.17
with an 1.60GHz Intel Pentium(M) processor and 512 MB of memory.
We measured the processing time of the Primality program
on different input parameters such as the number of attributes
and the number of FDs. The treewidth in all the test cases
was 3.
%  (i.e., at each tree node of the tree decomposition,
% there are maximal 4 elements (both 
% variables and clauses)).

\smallskip

\noindent
{\sc Test data generation.} Due to the lack of available
test data, we generated a balanced normalized tree decomposition.
Test data sets with increasing input parameters are then
generated by expanding the tree in a depth-first style. 
We have ensured that all different kinds of nodes
occur evenly in the tree decomposition.

\begin{table}
\begin{center}
%\begin{small}
\begin{tabular}{|c|r|r|r|r|r|}
\hline
tw&  \#Att& \#FD& \#tn& MD & MONA\\
\hline
3   & 3   &  1  &  3 &0.1& 650\\
3   & 6   &  2  &  12 &0.2& 9210\\
3   & 9   &  3  &  21 &0.4& 17930\\
3   & 12   &  4  &  34 &0.5& --\\
3   & 21  &  7  &  69 &0.8& --\\
3   & 33  &  11 & 105 &1.0& --\\
3   & 45  &  15 & 141 &1.2& --\\
3   & 57  &  19 & 193 &1.6& --\\
3   & 69  &  23 & 229 &1.8& --\\
3   & 81  &  27 & 265 &1.9& --\\
3   & 93  &  31 & 301 &2.2& --\\
\hline
\end{tabular}
\end{center}
%\end{small}
\caption{Processing Time in ms for PRIMALITY.}
\label{tab:experiments}
\end{table}
\smallskip

\noindent
{\sc Experimental results.} 
The outcome of the tests is shown in Table \ref{tab:experiments},
where tw stands for the 
treewidth;
\#Att, \#FD,  and \#tn 
stand for the number of attributes, 
FDs, and tree nodes, respectively.
The processing time (in ms) obtained with our
C++ implementation following the monadic datalog program
in Section~\ref{sec:primality} are
displayed in the column labelled ``MD''. The measurements
nicely reflect an essentially linear increase
of the processing time 
with the size of the input. Moreover, there is obviously no big
``hidden'' constant which would render the linearity useless.

In \cite{GPW06AAAI}, we proved the FPT of several non-monotonic 
reasoning problems via Courcelle's Theorem. Moreover, we 
also carried out some experiments with a prototype implementation 
using MONA (see \cite{mona}) for the MSO-model checking. 
We have now extended these experiments with MONA to the PRIMALITY 
problem. The time measurements of these experiments
are shown in the last column of Table~\ref{tab:experiments}. Due to problems discussed in \cite{GPW06AAAI}, MONA does not ensure linear
data complexity. Hence, all testes below line 3 of the table failed with
``out-of-memory'' errors. Moreover, also in cases where the exponential
data complexity does not yet ``hurt'',  our datalog approach outperforms the MSO-to-FTA approach by a factor of $1000$ or even more.

\smallskip

\noindent
{\sc Optimizations.} 
In our implementation, we have realized several
optimizations, which are highlighted below.
%which can be hardly achieved by the standard
%MSO-to-FTA approach.

%\begin{itemize}\item
{\em (1)  Succinct representation by non-monadic datalog.}
As was mentioned in the proofs of the Theorems~\ref{theo:three-col}
and \ref{theo:primtest}, our datalog programs 
can be regarded as succinct representations of 
big monadic datalog programs. If all possible ground instances
of our datalog rules had to be materialized, then 
we would end up with a ground program of the same size
as with the equivalent monadic program. However, it turns
out that the vast majority of possible instantiations is 
never computed since they are not ``reachable'' along the 
bottom-up computation.

{\em (2) General optimizations and lazy grounding.}
In principle, our implementation is based on the 
general idea of 
grounding followed by an evaluation of the ground program. This corresponds to 
the general technique to ensure linear time data complexity, 
cf.\ Theorem~\ref{theo:mondatalog-complexity}. 
A further improvement is achieved by the natural idea of
generating only those ground instances of rules which actually
produce new facts.

{\em (3)  Problem-specific optimizations of the non-monadic datalog
programs.}
In the discussion below Theorem~\ref{theo:three-col},
we have already mentioned
that the datalog programs presented in Section~\ref{sec:Work} 
incorporate several problem-specific optimizations. The underlying idea of these
optimizations is that 
many transitions which are kept track of by the 
generic construction in the proof of Theorem \ref{theo:MSO-to-Datalog}
(and, likewise, in the MSO-to-FTA approach) will not lead to 
a solution anyway. Hence, they 
are omitted in our
datalog programs right from the beginning.

{\em (4) Language extensions.} 
As was mentioned in Section~\ref{sec:Work}, we are using
language constructs (in particular, for handling sets of attributes
and FDs) which are not part of the datalog language. In principle, they could be realized in datalog. Nevertheless, we preferred an efficient
implementation of these constructs directly on C++ level. 
Further 
language 
extensions are conceivable and easy to realize.

{\em (5) Further improvements.} We are planning to implement further improvements. For instance, we are currently applying a strict 
bottom-up intuition as we compute new facts $\solve(v,\dots)$.
However, some top-down guidance in the 
style of magic sets so as not to compute all possible such 
facts at each level would be desirable.
Note that ultimately, at the root, 
only facts fulfilling certain conditions (like 
$a \not\in Y$, etc.) are needed in case that an attribute
$a$ is indeed prime.

\section{Conclusion} 
\label{sec:Conclusion} 

In this work, we have proposed a new approach based on monadic datalog 
to tackle a big class of fixed-parameter tractable problems. Theoretically, we
have shown that 
every MSO-definable unary query over finite structures with bounded
treewidth is also definable in monadic datalog. In fact, the 
resulting program even lies in a particularly efficient fragment 
of monadic datalog.
Practically, we have put this approach to work by applying it 
to 
the 3-Colorability problem and the 
PRIMALITY problem  with 
bounded treewidth. The experimental results 
thus obtained look very promising. They underline that
datalog with its potential for optimizations and its flexibility is
clearly worth considering for this class of problems.

Recall that the PRIMALITY problem is closely related to 
an important problem  in the area of artificial 
intelligence,  namely the relevance problem of 
propositional abduction (i.e., given a system description 
in form of a propositional clausal theory and observed symptoms, one
has to decide if some hypothesis is part of a possible
explanation of the symptoms). Indeed, if the clausal theory
is restricted to definite Horn clauses and if we are only interested in
minimal explanations, then the relevance problem is basically
the same as the problem of deciding primality in a subschema
$R' \subseteq R$. Extending our $\primality()$ program (and, in particular,
the $\solve()$-predicate) 
from Section~\ref{sec:Work}
so as to test primality in a subschema is  rather straightforward. On the other hand, extending such a program to 
abduction with arbitrary clausal theories (which is on the second level
of the polynomial hierarchy, see \cite{Eiter-Gottlob-Abduction-JACM95})
is much more involved. A monadic datalog program solving the relevance
problem also in this general case was presented in 
\cite{GPW08AAAI}.

Our datalog program in Section \ref{sec:Work} 
was obtained by an ad hoc construction rather than via a generic transformation from
MSO. Nevertheless, we are convinced that the idea of a bottom-up propagation of 
certain conditions is quite generally applicable. We are therefore 
planning to tackle many more problems, whose FPT was established via Courcelle's Theorem, with this new approach. 
We have already incorporated some optimizations into our implementation.
Further improvements are on the way 
(in particular, further heuristics to prune irrelevant
parts of the search space).

\bibliographystyle{abbrv}

%\bibliography{monadic}

%%%%%%%%%%%%%%%%%%%%%%%%%%%%%%%%%%%%%%%%%%%%%%%%%%%%%%%%%%%%%%%%%%%%%%

%\small

%\nop{******************

% ******************}%nop

\end{document}